\ifdefined\renderlayout
\else
    \def\renderlayout{arXiv}
\fi

\def\renderlayoutACM{ACM}
\def\renderlayoutACManon{ACManon}
\def\renderlayoutarXiv{arXiv}
\def\renderlayoutarXivanon{arXivanon}

\def\classoptions{}
\ifx \renderlayout\renderlayoutACM
    \def\classoptions{acm}
\else
    \ifx \renderlayout\renderlayoutACManon
        \def\classoptions{acm,anonymous}
    \else
        \ifx \renderlayout\renderlayoutarXiv
            \def\classoptions{arxiv}
        \else
            \ifx \renderlayout\renderlayoutarXivanon
                \def\classoptions{arxiv,anonymous}
            \else
                \PackageError{NoPackage}{Invalid package layout}
            \fi
        \fi
    \fi
\fi

\documentclass[\classoptions]{online-local}

\usepackage{algorithm2e}
\usepackage{amsmath}
\usepackage{amsthm}
\usepackage{booktabs}
\usepackage{doi}
\usepackage{enumitem}
\usepackage[utf8]{inputenc}
\usepackage{graphicx}
\usepackage{microtype}
\usepackage{thmtools}
\usepackage{thm-restate}
\usepackage{tikz-cd}
\usepackage{url}
\usepackage{accents}
\usepackage[capitalise,noabbrev]{cleveref}
\usepackage{xspace}
\usepackage{import}
\usepackage{afterpage}
\usepackage{xcolor}

\metaCopyright{none}
\metaCopyrightYear{2022}
\metaYear{2022}
\metaDOI{}
\metaConference[]{}{}{}
\metaBooktitle{}
\metaPrice{}
\metaISBN{}
\title{Locality in online, dynamic, sequential, and distributed graph algorithms}

\makeatletter
\def\namedlabel#1#2{\begingroup
    #2%
    \def\@currentlabel{#2}%
    \phantomsection\label{#1}\endgroup
}
\makeatother

\newcommand{\flushfigures}{\afterpage{\clearpage}}

\newcommand{\fakeparagraph}[2]{\par\noindent\textbf{#1}\hspace{1em}#2}
\newcommand{\pequiv}{\mathrel{\accentset{\star}{\sim}}}

\theoremstyle{\acmorarxiv{acmplain}{plain}}
\newtheorem{theorem}{Theorem}[section]
\newtheorem{lemma}[theorem]{Lemma}
\newtheorem{corollary}[theorem]{Corollary}
\newtheorem{observation}[theorem]{Observation}

\theoremstyle{\acmorarxiv{acmdefinition}{definition}}
\newtheorem{definition}[theorem]{Definition}

\theoremstyle{\acmorarxiv{acmdefinition}{definition}}

\newtheorem*{question*}{Question}

\newtheorem*{hypothesis*}{Hypothesis}

\crefname{line}{line}{lines}

\newcommand{\simplemodelop}[1]{\ifmmode\operatorname{#1}\else#1\fi\xspace}

\newcommand{\SLOCAL}{\simplemodelop{SLOCAL}}
\newcommand{\LOCAL}{\simplemodelop{LOCAL}}
\newcommand{\onlineLOCAL}{\simplemodelop{online-LOCAL}}
\newcommand{\dynamicLOCAL}{\simplemodelop{dynamic-LOCAL}}
\newcommand{\randomizedLOCAL}{\simplemodelop{randomized-LOCAL}}
\newcommand{\OnlineLOCAL}{\simplemodelop{Online-LOCAL}}
\newcommand{\DynamicLOCAL}{\simplemodelop{Dynamic-LOCAL}}

\DeclareMathOperator{\class}{Class}
\DeclareMathOperator{\dist}{dist}
\DeclareMathOperator{\poly}{poly}
\newcommand{\algoa}{\mathcal{A}}
\newcommand{\algob}{\mathcal{B}}

\makeatletter
\DeclareMathOperator{\regularlog}{log}
\renewcommand{\log}{\protect\@ifstar{\regularlog^*}{\regularlog}}
\makeatother

\newcommand{\pathform}[1]{#1^{\text{path}}}
\newcommand{\machine}{\mathcal{M}}
\newcommand{\twohalf}{2.5}
\newcommand{\probtwohalf}{\Pi_{\twohalf}}

\definecolor{mygreen}{HTML}{A5CEBA}
\definecolor{myblue}{HTML}{7F9CAB}
\colorlet{myred}{red!80!black}

\begin{document}

\myauthor{Amirreza Akbari}{amirreza.akbari@aalto.fi}{Aalto University}{Espoo}{Finland}
\myauthor{Navid Eslami}{navid.eslami@aalto.fi}{Aalto University}{Espoo}{Finland}
\myauthor{Henrik Lievonen}{henrik.lievonen@aalto.fi}{Aalto University}{Espoo}{Finland}
\myauthor{Darya Melnyk}{darya.melnyk@aalto.fi}{Aalto University}{Espoo}{Finland}
\myauthor{Joona Särkijärvi}{joona.sarkijarvi@aalto.fi}{Aalto University}{Espoo}{Finland}
\myauthor{Jukka Suomela}{jukka.suomela@aalto.fi}{Aalto University}{Espoo}{Finland}

\begin{abstract}
In this work, we give a unifying view of locality in four settings: distributed algorithms, sequential greedy algorithms, dynamic algorithms, and online algorithms.

We introduce a new model of computing, called the \onlineLOCAL model: the adversary reveals the nodes of the input graph one by one, in the same way as in classical online algorithms, but for each new node we get to see its radius-$T$ neighborhood before choosing the output. Instead of \emph{looking ahead} in time, we have the power of \emph{looking around} in space.

We compare the \onlineLOCAL model with three other models: the \LOCAL model of distributed computing, where each node produces its output based on its radius-$T$ neighborhood, its sequential counterpart \SLOCAL, and the \dynamicLOCAL model, where changes in the dynamic input graph only influence the radius-$T$ neighborhood of the point of change.

The \SLOCAL and \dynamicLOCAL models are sandwiched between the \LOCAL and \onlineLOCAL models, with \LOCAL being the weakest and \onlineLOCAL the strongest model. In general, all four models are distinct, but we study in particular \emph{locally checkable labeling problems} (LCLs), which is a family of graph problems extensively studied in the context of distributed graph algorithms.

We prove that for LCL problems in paths, cycles, and rooted trees, all four models are roughly equivalent: the locality of any LCL problem falls in the same broad class---$O(\log^* n)$, $\Theta(\log n)$, or $n^{\Theta(1)}$---in all four models. In particular, this result enables one to generalize prior lower-bound results from the \LOCAL model to all four models, and it also allows one to simulate e.g.\ \dynamicLOCAL algorithms efficiently in the \LOCAL model.

We also show that this equivalence does not hold in two-dimensional grids or general bipartite graphs. We provide an \onlineLOCAL algorithm with locality $O(\log n)$ for the $3$-coloring problem in bipartite graphs---this is a problem with locality $\Omega(n^{1/2})$ in the \LOCAL model and $\Omega(n^{1/10})$ in the \SLOCAL model.
\end{abstract}

\maketitle

\section{Introduction}\label{sec:intro}

In \emph{online graph algorithms}, the adversary reveals the graph one node at a time, and the algorithm has to respond by labeling the node based on what it has seen so far. For example, in \emph{online graph coloring}, we need to pick a color for the node that was revealed, in such a way that the end result is a proper coloring of the graph.

Usually, when the adversary reveals a node $v$, the algorithm gets to see only the edges incident to of $v$. In this work, we consider a more general setting: the algorithm gets to see the radius-$T$ neighborhood of $v$, i.e., it can \emph{look around} further in the input graph. 
For $T=0$, this model corresponds to the typical online model. For $T=n$, on the other hand, it is clear that any graph problem (in connected graphs) is solvable in this setting. The key question is what value of $T$ is sufficient for a given graph problem. Put otherwise, what is the \emph{locality} of a given online problem?

It turns out that this question is very closely connected to questions studied in the context of \emph{distributed} graph algorithms, and we can identify problem classes in which the online setting coincides with the distributed setting. However, we will also see surprising differences: the prime example will be the problem of $3$-coloring bipartite graphs, which is a fundamentally \emph{global} problem in the distributed setting, while we show that we can do much better in the online setting.

\begin{figure}
\newcommand{\mynote}[1]{{\color{myblue}\small#1}}
\newcommand{\mycols}{@{\hspace{0.5ex}}c@{\hspace{0.5ex}}}
\newcommand{\mybox}[1]{\fbox{#1}}
\centering
\begin{tikzcd}[column sep=-5mm, row sep=-3mm]
&
\hspace{22mm}
&
&
    \mybox{\begin{tabular}{\mycols}
        \text{\SLOCAL}\\[2pt]
        \mynote{distributed,}\\[-2pt]
        \mynote{sequential}
    \end{tabular}}
    \arrow[rrrd, "\subsetneq"]
    \arrow[rddd, leftrightarrow, color=myred, "\nsubseteq"]
    \arrow[lddd, leftrightarrow, color=myred, "\nsubseteq"]
&
&
\hspace{22mm}
&
\\
    \mybox{\begin{tabular}{\mycols}
        \text{\LOCAL}\\[2pt]
        \mynote{distributed,}\\[-2pt]
        \mynote{parallel}
    \end{tabular}}
    \arrow[rrru, "\subsetneq"]
    \arrow[rrdd, "\subsetneq"]
&
&
&
&
&
&
    \mybox{\begin{tabular}{\mycols}
        online-\\
        \LOCAL\\[2pt]
        \mynote{centralized}
    \end{tabular}}
\\
\rule{0pt}{10mm}
\\
&
&
    \mybox{\begin{tabular}{\mycols}
        dynamic-\\
        $\LOCAL^\pm$\\[2pt]
        \mynote{centralized}
    \end{tabular}}
    \arrow[rr, "\subsetneq"]
&
&
    \mybox{\begin{tabular}{\mycols}
        dynamic-\\
        \LOCAL\\[2pt]
        \mynote{centralized}
    \end{tabular}}
    \arrow[rruu, "\subsetneq"]
&
&
    \mybox{\begin{tabular}{\mycols}
        online graph\\
        algorithms\\[2pt]
        \mynote{centralized}
    \end{tabular}}
    \arrow[uu, "\subsetneq"]
\end{tikzcd}
\caption{The landscape of models.}\label{fig:diagram}
\end{figure}
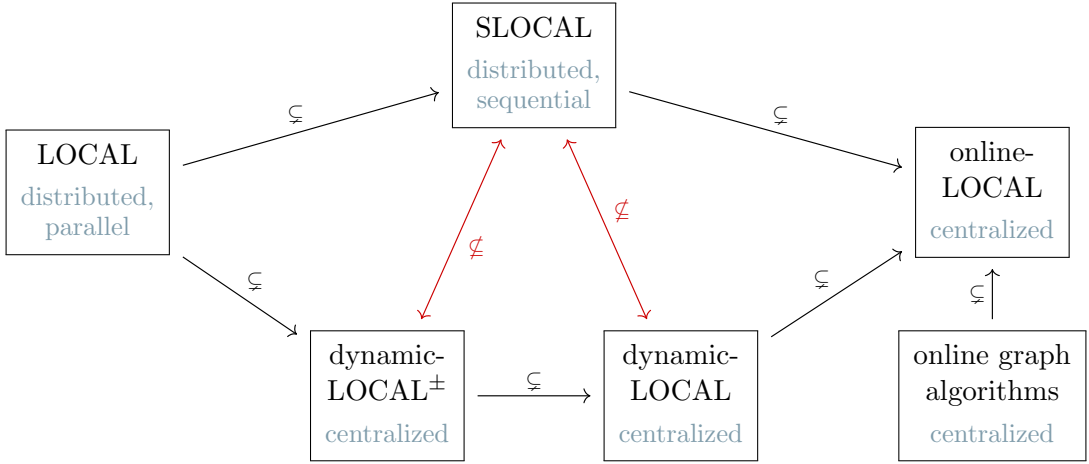

\subsection{Contribution 1: landscape of models}

In \cref{sec:def-related}, we define the new model of online algorithms with lookaround, which we call \onlineLOCAL, and we also recall the definitions of three models familiar from the fields of distributed and dynamic graph algorithms:
\begin{itemize}
    \item The \LOCAL model \cite{linial92,Peleg2000}: the nodes are processed \emph{simultaneously} in parallel; each node looks at its radius-$T$ neighborhood and picks its own output.
    \item The \SLOCAL model \cite{slocal}: the nodes are processed \emph{sequentially} in an adversarial order; each node in its turn looks at its radius-$T$ neighborhood and picks its own output (note that here the output of a node may depend on the outputs of other nodes that were previously processed).
    \item The \dynamicLOCAL model: the adversary \emph{constructs} the graph by adding nodes and edges one by one; after each modification, the algorithm can only update the solution within the radius-$T$ neighborhood of the point of change.
    While this is not one of the standard models, there is a number of papers \cite{10.1007/3-540-57899-4_44,10.1145/2700206,10.1145/3188745.3188922,Barenboim17,DBLP:journals/corr/abs-1804-01823,doi:10.1137/1.9781611975031.1,du2018improved} that implicitly make use of this model.
    We will also occasionally consider the $\dynamicLOCAL^\pm$ model, in which we can have both additions and deletions.
\end{itemize}
In \cref{sec:landscape}, we show that we can sandwich \SLOCAL and both versions of \dynamicLOCAL between \LOCAL and \onlineLOCAL, as shown in \cref{fig:diagram}. In particular, this implies that if we can prove that \LOCAL and \onlineLOCAL are equally expressive for some family of graph problems, we immediately get the same result also for \SLOCAL and \dynamicLOCAL. This is indeed what we will achieve in our next contribution.

\subsection{Contribution 2: collapse for LCLs in rooted regular trees}\label{ssec:contribution2}

A lot of focus in the study of distributed graph algorithms and the \LOCAL model has been on understanding \emph{locally checkable labeling problems} (in brief, LCLs)~\cite{balliu19lcl-decidability,balliu20binary-labeling,chang21automata-theoretic,naor95,brandt17grid-lcl,Chang2017ATH,balliu19weak-col,balliu21rooted-trees,chang2020complexity}. These are problems where feasible solutions are defined with local constraints---a solution is feasible if it looks good in all constant-radius neighborhoods (see \cref{def:lcl}). Coloring graphs of maximum degree $\Delta$ with $k$ colors (for some constants of $\Delta$ and $k$) is an example of an LCL problem.

\begin{table}
\newcommand{\myeq}{$\Leftrightarrow$}
\centering
\begin{tabular}{@{}l@{\quad\qquad}lclclcl@{}}
\toprule
& \multicolumn{2}{@{}l}{\LOCAL}
& \multicolumn{2}{l}{\SLOCAL}
& \multicolumn{2}{l}{dynamic-}
& online-
\\
&&
&&
& \multicolumn{2}{l}{\LOCAL}
& \LOCAL
\\
\midrule
LCLs in paths and cycles
& $O(\log^* n)$ &\myeq& $O(1)$ &\myeq& $O(1)$ &\myeq& $O(1)$ \\
& $\Theta(n)$ &\myeq& $\Theta(n)$ &\myeq& $\Theta(n)$ &\myeq& $\Theta(n)$ \\
\midrule
LCLs in rooted regular trees
& $O(\log^* n)$ &\myeq& $O(1)$ &\myeq& $O(1)$ &\myeq& $O(1)$ \\
& $\Theta(\log n)$ &\myeq& $\Theta(\log n)$ &\myeq& $\Theta(\log n)$ &\myeq& $\Theta(\log n)$ \\
& $n^{\Theta(1)}$ &\myeq& $n^{\Theta(1)}$ &\myeq& $n^{\Theta(1)}$ &\myeq& $n^{\Theta(1)}$ \\
\bottomrule
\end{tabular}
\caption{In all four models, LCL problems have got the same locality classes in paths, cycles, and rooted trees. Here $n^{\Theta(1)}$ refers to locality $\Theta(n^\alpha)$ for some constant $\alpha > 0$. See \cref{sec:collapse} for more details.}\label{tab:classes}
\end{table}

In \cref{sec:collapse}, we study LCL problems in \emph{paths}, \emph{cycles}, and \emph{rooted regular trees}, and we show that all four models are approximately equally strong in these settings:
\begin{equation}\label{eq:tree-approx}
    \LOCAL \approx \SLOCAL \approx \dynamicLOCAL \approx \onlineLOCAL.
\end{equation}
For example, we show that if the locality of an LCL problem in rooted trees is $n^{\Theta(1)}$ in the \LOCAL model, it is also $n^{\Theta(1)}$ in the \dynamicLOCAL, \SLOCAL, and \onlineLOCAL models. Refer to \cref{tab:classes} for the full classification.

By previous work, we know that LCL complexities in paths, cycles, and rooted regular trees are \emph{decidable} in the \LOCAL model \cite{balliu21rooted-trees,chang21automata-theoretic,balliu19lcl-decidability}. Our equivalence result allows us to extend this decidability to the \SLOCAL, \dynamicLOCAL, and \onlineLOCAL models. For example, there is an algorithm that gets as input the description of an LCL problem in rooted trees and produces as output in which of the classes of \cref{tab:classes} it is, for any of the four models.

\subsection{Contribution 3: 3-coloring bipartite graphs in \onlineLOCAL}\label{ssec:contribution3}

Given the equivalence results for LCLs in paths, cycles, and rooted regular trees, it would be tempting to conjecture that the models are approximately equal for LCLs in any graph class.
In \cref{sec:bipartite-separation}, we show that this is not the case:
we provide an exponential separation between the \SLOCAL and \onlineLOCAL models for the problem of $3$-coloring bipartite graphs. By prior work it is known that in the \LOCAL model, the locality of $3$-coloring is $\Omega(n^{1/2})$ in two-dimensional grids \cite{brandt17grid-lcl}, which are a special case of bipartite graphs; using this result we can derive a lower bound of $\Omega(n^{1/10})$ also for the \SLOCAL model (see \cref{app:slocal-3-coloring}). In \cref{sec:bipartite-separation}, we prove the following:
\begin{restatable}{theorem}{bipartitecoloring}\label{thm:bipartitecoloring}
    There is an \onlineLOCAL algorithm that finds a $3$-coloring in bipartite graphs with locality $O(\log n)$.
\end{restatable}
That is, in bipartite graphs, there is an LCL problem that requires locality $n^{\Omega(1)}$ in the \LOCAL and \SLOCAL models and is solvable with locality $O(\log n)$ in the \onlineLOCAL model.

The algorithm that we will present for coloring bipartite graph is also interesting from the perspective of competitive analysis of online algorithms. With locality $O(\log n)$, the \onlineLOCAL algorithm can compute a $3$-coloring. Since bipartite graphs are $2$-colorable, this gives us a $1.5$-competitive \onlineLOCAL algorithm. On the other hand, it has been shown that any online algorithm for coloring bipartite graphs is at least $\Omega(\log n)$-competitive~\cite{10.2307/2272247}, with a matching algorithm presented in \cite{LOVASZ1989319}. This result shows how much the competitive ratio of an algorithm can be improved by increasing the view of each node.

\subsection{Contribution 4: locality of online coloring}

As a corollary of our work, together with results on distributed graph coloring from prior work \cite{linial92,brandt17grid-lcl,cole86deterministic}, we have now a near-complete understanding of the locality of graph coloring in paths, cycles, rooted trees, and grids in both distributed and online settings.
\Cref{tab:coloring} summarizes our key results.
For the proofs of the localities in the \onlineLOCAL model, see \cref{sec:bipartite-separation,sec:collapse}.

\begin{table}
    \centering
    \begin{tabular}{@{}l@{\quad\quad}llllll@{}}
    \toprule
    & colors & competitive & \LOCAL & \SLOCAL & online- & references \\
    && ratio &&& \LOCAL & \\
    \midrule
    Rooted trees
    & $2$ & $1$ & $\Theta(n)$ & $\Theta(n)$ & $\Theta(n)$ & trivial \\
    and paths
    & $3$ & $1.5$ & $\Theta(\log^* n)$ & $O(1)$ & $O(1)$ & \cite{linial92,cole86deterministic} \\
    & $4$ & $2$ & $\Theta(\log^* n)$ & $O(1)$ & $O(1)$ & \cite{linial92,cole86deterministic} \\
    & \ldots \\
    \midrule
    Grids & $2$ & $1$ & $\Theta(n^{1/2})$ & $\Theta(n^{1/2})$ & $\Theta(n^{1/2})$ & trivial \\
    & $3$ & $1.5$ & $\Theta(n^{1/2})$ & \boldmath $\Omega(n^{1/10})$ & \boldmath $O(\log n)$ & \cref{sec:bipartite-separation}, \cite{brandt17grid-lcl} \\
    & $4$ & $2$ & $\Theta(\log^* n)$ & $O(1)$ & $O(1)$ & \cite{brandt17grid-lcl} \\
    & $5$ & $2.5$ & $\Theta(\log^* n)$ & $O(1)$ & $0$ & \cite{brandt17grid-lcl} \\
    & \ldots \\
    \bottomrule
    \end{tabular}
    \caption{The locality of the vertex coloring problem in distributed vs.\ online settings, for two graph families: rooted trees and paths (with $n$ nodes) and $2$-dimensional grids (with $\sqrt{n} \times \sqrt{n}$ nodes). See \cref{sec:bipartite-separation,sec:collapse,app:slocal-3-coloring} for more details.}
    \label{tab:coloring}
\end{table}

\subsection{Motivation}

Before we discuss the key technical ideas, we will briefly explain the practical motivation for the study of \onlineLOCAL and \dynamicLOCAL models. As a running example, consider the challenge of providing public services (e.g.\ local schools) in a rapidly growing city. The future is unknown, depending on future political decisions, yet the residents will need services every day.

The offline solution would result in a city-wide redesign of e.g.\ the entire school network every time the city plan is revised; this is not only costly but also disruptive. On the other hand, a strict online solution without any consideration of the future would commit to a solution that is far from optimal. The models that we study in this work capture the essence of two natural strategies for coping with such a situation:
\begin{itemize}
    \item Redesign the public service network only in the local neighborhoods in which there are new developments. This corresponds to the \dynamicLOCAL model, and the locality parameter $T$ captures cost the redesign cost and the disruption it causes.
    \item Wait until new developments in a neighborhood are completed before providing permanent public services in the area. This corresponds to the \onlineLOCAL model, and locality $T$ captures the inconvenience of the residents (e.g., distance to the nearest school).
\end{itemize}
These two models make it possible to formally explore trade-offs between the quality of the solution in the long term vs.\ the inconvenience of those living close to the areas where the city is changing. In these kinds of scenarios the key challenge is not related to the computational cost of finding an optimal solution (which is traditionally considered in the context of dynamic graph algorithms) but to the quality of the solution (which is typically the focus in online algorithms). The key constraint is not the availability of information on the current state of the world (which is traditionally considered in distributed graph algorithms), but the cost of changing the solution.

\subsection{Techniques and key ideas}

For the equivalence \eqref{eq:tree-approx} in paths and cycles (\cref{ssec:collapse-paths-cycles}), we first make use of \emph{pumping-style arguments} that were introduced by \citet*{Chang2017ATH} in the context of distributed algorithms. We show that such ideas can be used to also analyze locality in the context of online algorithms: we start by showing that we can ``speed up'' (or ``further localize'') \onlineLOCAL algorithms with a \emph{sublinear} locality to \onlineLOCAL algorithms with a \emph{constant} locality in paths and cycles. Then, once we have reached constant locality in the \onlineLOCAL model, we show how to turn it into a \LOCAL-model algorithm with locality $O(\log^* n)$. In this part, the key insight is that \emph{we cannot directly simulate} \onlineLOCAL in \LOCAL. Instead, we can use an \onlineLOCAL algorithm with a constant locality to find a \emph{canonical labeling} for each possible input-labeled fragment, and use this information to design a \LOCAL-model algorithm. The main trick is that we first present only disconnected path fragments to an \onlineLOCAL algorithm, and force it to commit to some output labeling in each fragment without knowing how the fragments are connected to each other.

In the case of rooted regular trees (\cref{ssec:collapse-rooted-trees}), we face the same fundamental challenge: we cannot directly simulate black-box \onlineLOCAL algorithms in the \LOCAL model. Instead, we will need to look at the combinatorial properties of a given LCL problem $\Pi$. We proceed in two steps: (1)~Assume that the locality of $\Pi$ is $n^{\Theta(1)}$ in the \LOCAL model; we need to show that the locality is $n^{\Theta(1)}$ also in the \onlineLOCAL model. Using the result of \cite{balliu21rooted-trees}, high \LOCAL-model locality implies that the structure of $\Pi$ has to have certain ``inflexibilities'', and we use this property to present a strategy that the adversary can use to force any \onlineLOCAL algorithm with locality $n^{o(1)}$ to fail. (2)~Assume that we have an \onlineLOCAL algorithm $A$ for $\Pi$ with locality $o(\log n)$; we need to show that the locality is $O(\log^* n)$ in the \LOCAL model. Here we design a family of inputs and a strategy of the adversary that forces algorithm $A$ to construct a ``certificate'' (in the sense of \cite{balliu21rooted-trees}) that shows that $\Pi$ is efficiently solvable in the \LOCAL model.

For $3$-coloring bipartite graphs in \onlineLOCAL (\cref{sec:bipartite-separation}), we make use of the following ideas. We maintain a collection of graph fragments such that each of the fragments has got a boundary that is properly $2$-colored. Each such fragment has got one of two possible parities (let us call them here ``odd'' and ``even'') with respect to the underlying bipartition. We do not know the global parity of a given graph fragment until we have seen almost the entire graph. Nevertheless, it is possible to merge two fragments and maintain the invariant: if two fragments $A$ and $B$ have parities that are not compatible with each other, we can surround either $A$ or $B$ with a \emph{barrier} that uses the third color, and thus change parities. Now we can merge $A$ and $B$ into one fragment that has got a properly $2$-colored boundary. The key observation here is that we can make a \emph{choice} between surrounding $A$ vs.\ $B$, and if we always pick the one with the smallest number of nested barriers, we will never need to use more than a logarithmic number of nested barriers. It turns out that this is enough to ensure that seeing up to distance $O(\log n)$ suffices to color any node chosen by the adversary.

\subsection{Open questions}

Our work gives rise to a number of open questions. First, we can take a more fine-grained view of the results in \cref{tab:classes,tab:coloring}:
\begin{enumerate}
    \item Is there any problem in rooted trees with locality $\Theta(n^{\alpha})$ in the \onlineLOCAL model and locality $\Theta(n^{\beta})$ in the \LOCAL model, for some $\alpha < \beta$?
    \item Is it possible to find a $3$-coloring in $2$-dimensional graphs in the \dynamicLOCAL model with locality $O(\log n)$?
    \item Is it possible to find a $3$-coloring in bipartite graphs in the \onlineLOCAL model with locality $o(\log n)$?
\end{enumerate}
Perhaps even more interesting is what happens if we consider \emph{unrooted} trees instead of rooted trees. In unrooted trees we can separate \emph{randomized} and \emph{deterministic} versions of the \LOCAL model \cite{chang16}, and \SLOCAL is strong enough to derandomize randomized \LOCAL-model algorithms \cite{DerandomizingLocal}; hence the key question is:
\begin{enumerate}[resume]
    \item Does $\randomizedLOCAL \approx \SLOCAL \approx \dynamicLOCAL \approx \onlineLOCAL$ hold for LCL problems in unrooted trees?
\end{enumerate}
Finally, our work shows a trade-off between the competitive ratio and the locality of coloring: With locality $O(\log n)$, one can achieve $O(1)$-coloring of a bipartite graph, and to achieve locality $0$, one needs to use $\Omega(\log n)$ colors. This raises the following question:
\begin{enumerate}[resume]
    \item What tradeoffs exist between the locality and number of colors needed to color a (bipartite) graph in the \onlineLOCAL model?
\end{enumerate}

\section{Definitions and related work}\label{sec:def-related}

Throughout this work, graphs are simple, undirected, and finite, unless otherwise stated. We write $G = (V,E)$ for a graph $G$ with the set of nodes $V$ and the set of edges $E$, and we use $n$ to denote the number of nodes in the graph. For a node $v$ and a natural number $T$, we will use $B(v,T)$ to denote the set of all nodes in the radius-$T$ neighborhood of node $v$. For a set of nodes $U$, we write $G[U]$ for the subgraph of $G$ induced by $U$. By radius-$T$ neighborhood of $v$ we refer to the induced subgraph $G[B(v,T)]$, together with possible input and output labelings.

We will use the following notation for graph problems. We write $\mathcal{G}$ for the family of graphs, $\Sigma$ for the set of input labels, and $\Gamma$ for the set of output labels. For a graph $G = (V,E)$, we write $I\colon V \to \Sigma$ for the input labeling and $L\colon V \to \Gamma$ for the output labeling. We consider here node labelings, but edge labelings can be defined in an analogous manner. A \emph{graph problem} $\Pi$ associates with each possible \emph{input} $(G,I)$ a set of feasible \emph{solutions} $L$; this assignment must be invariant under graph isomorphism.

\paragraph{Locality.}

In what follows, we will define five models of computing: \LOCAL, \SLOCAL, two versions of \dynamicLOCAL, and \onlineLOCAL. In all of these models, an algorithm is characterized by a \emph{locality} $T$ (a.k.a.\ locality radius, local horizon, time complexity, or round complexity, depending on the context). In general, $T$ can be a function of $n$. We will assume that the algorithm knows the value of $n$.

In each of these models $\mathcal{M}$, we say that algorithm $\algoa$ solves problem $\Pi$ if for each possible input $(G,I)$ and for each possible adversarial choice, the labeling $L$ produced by $\algoa$ is a feasible solution. We say that problem $\Pi$ has locality $T$ in model $\mathcal{M}$ if $T$ is the pointwise smallest function such that there exists an $\mathcal{M}$-model algorithm $\algoa$ that solves $\Pi$ with locality at most $T$.

\paragraph{\boldmath \LOCAL model.}

In the \LOCAL model of distributed computing \cite{linial92,Peleg2000}, the adversary labels the nodes with unique identifiers from $\{1,2,\dotsc,\poly(n)\}$. In a \LOCAL model algorithm, each node in parallel chooses its local output based on its radius-$T$ neighborhood (the output may depend on the graph structure, input labels, and the unique identifiers).

\citet*{naor95} initiated the study of the locality of LCL problems (see \cref{def:lcl}) in the \LOCAL model. Today, LCL problems are well classified with respect to their locality for the special cases of paths~\cite{balliu19lcl-decidability,balliu20binary-labeling,brandt17grid-lcl,chang21automata-theoretic,naor95}, grids~\cite{brandt17grid-lcl}, directed and undirected trees~\cite{Chang2017ATH,balliu20binary-labeling,balliu19weak-col,balliu21rooted-trees,chang2020complexity} as well as general graphs~\cite{naor95,brandt17grid-lcl}, with only a few unknown gaps~\cite{balliu21rooted-trees}.

\paragraph{\boldmath \SLOCAL model.}

In the \SLOCAL model \cite{slocal}, we have got adversarial unique identifiers similar to the \LOCAL model, but the nodes are processed sequentially with respect to an adversarial input sequence $ \sigma = v_1,v_2,v_3,\ldots, v_n$. Each node $v$ is equipped with an unbounded local memory; initially, all local memories are empty. When a node $v$ is processed, it can query the local memories of the nodes in its radius-$T$ neighborhood, and based on this information, it has to decide what is its own final output and what to store in its own local memory.

The \SLOCAL model has been used as a tool to e.g.\ better understand the role of randomness in the \LOCAL model \cite{slocal,DerandomizingLocal}. It is also well-known that \SLOCAL is strictly stronger than \LOCAL. For example, it is trivial to find a maximal independent set greedily in the \SLOCAL model, while this is a nontrivial problem in the general case in the \LOCAL model \cite{kuhn16local,balliu21mm}. There are many LCL problems with \LOCAL-locality $\Theta(\log^* n)$ \cite{linial92,cole86deterministic}, and all of them have \SLOCAL-locality $O(1)$. There are also LCL problems (e.g.\ the so-called \emph{sinkless orientation} problem), where the locality in the (deterministic) \LOCAL model is $\Theta(\log n)$, while the locality in the (deterministic) \SLOCAL model is $\Theta(\log \log n)$ \cite{chang16,DerandomizingLocal}.

\paragraph{\boldmath \DynamicLOCAL model.}

To our knowledge, there is no standard definition or name for what we call \dynamicLOCAL here; however, the idea has appeared implicitly in a wide range of work. For example, many efficient dynamic algorithms for graph problems such as vertex or edge coloring, maximal independent set, and maximal matching also satisfy the property that the solution is only modified in the (immediate) local neighborhood of a point of change \cite{10.1007/3-540-57899-4_44,10.1145/2700206,10.1145/3188745.3188922,Barenboim17,DBLP:journals/corr/abs-1804-01823,doi:10.1137/1.9781611975031.1,du2018improved}, and hence all of them fall in the class \dynamicLOCAL.

We use the following definition for \dynamicLOCAL: Computation starts with an empty graph $G_0$. In step $i$, the adversary constructs a supergraph $G_i$ of $G_{i-1}$ such that $G_i$ and $G_{i-1}$ differ in only one edge or one node; let $C_i$ denote the set of nodes $v$ in $G_i$ with $G_i[B(v,T)] \ne G_{i-1}[B(v,T)]$, i.e., nodes that are within distance at most $T$ from the point of change. In each step, the algorithm has to produce a feasible labeling $L_i$ for problem $\Pi$ in graph $G_i$, and the labeling can only be modified in the local neighborhood of a point of change, i.e., $L_i(v) = L_{i-1}(v)$ for all $v \notin C_i$.

Note that we defined the \dynamicLOCAL model for the \emph{incremental} case, where nodes and edges are only added.
If we do not require that $G_i$ is a supergraph of $G_{i-1}$, we arrive at what we call the $\dynamicLOCAL^\pm$ model with both additions and deletions.

\paragraph{Online graph algorithms.}

In online graph algorithms, nodes are processed sequentially with respect to an adversarial input sequence $\sigma = v_1,v_2,\dotsc,v_n$. Let $\sigma_i = v_1,v_2,\dotsc,v_i$ denote the first $i$ nodes of the sequence, and let $G_i = G[\{v_1,v_2,\dotsc,v_i\}]$ be the subgraph induced by these nodes. When the adversary presents a node $v_i$, the algorithm has to label $v_i$ based on  $\sigma_i$ and $G_i$.

Online algorithms on graphs have been studied for many problems such as matching~\cite{onlineMatching} and independent set~\cite{10.1016/S0304-3975(01)00411-X}, but closest to our work is the extensive literature on online graph coloring \cite{10.2307/2272247,albers2021tight,doi:10.1002/jgt.3190120212,HALLDORSSON1994163,HALLDORSSON1997265,LOVASZ1989319,VISHWANATHAN1992657}. There is also prior work that has considered various ways to strengthen the notion of online algorithms; the performance of online algorithms can be improved by letting the algorithm know the input graph~\cite{OnlineColoringKnownGraphs,10.1007/978-3-642-38016-7_2}, by giving it an advice string~\cite{onlineProblemsAdvice,10.1007/978-3-642-32241-9_44,10.1007/978-3-662-49192-8_19} with knowledge about the request sequence, or allowing the algorithm to delay decisions~\cite{onlineMatchingDelays}. The \onlineLOCAL model can be interpreted as online graph algorithms with spatial advice, and it can also be interpreted as a model where the online algorithm can delay its decision for node~$v$ until it has seen the whole neighborhood around $v$ (this interpretation is equivalent to the definition we give next).

\paragraph{\boldmath \OnlineLOCAL model.}

We define the \onlineLOCAL model as follows. The nodes are processed sequentially with respect to an adversarial input sequence $\sigma = v_1,v_2,\dotsc,v_n$. Let $\sigma_i = v_1,v_2,\dotsc,v_i$ denote the first $i$ nodes of the sequence, and let
$
    G_i = G\bigl[\bigcup_{j=1}^{i}B(v_j,T)\bigr]
$
be the subgraph induced by the radius-$T$ neighborhoods of these nodes. When the adversary presents a node $v_i$, the algorithm has to label $v_i$ based on  $\sigma_i$ and $G_i$.

Observe that any online graph algorithm is an \onlineLOCAL algorithm with locality $0$. Further note that in the \onlineLOCAL model, unique identifiers would not give any additional information. This is because the nodes can always be numbered with respect to the point in time when the algorithm first sees the them in some $G_i$.

Yet another way to interpret the \onlineLOCAL model is that it is an extension of the \SLOCAL model, where the algorithm is equipped with unbounded \emph{global memory} where it can store arbitrary information on what has been revealed so far. When they introduced the \SLOCAL model, \citet*{slocal} mentioned the possibility of such an extension but pointed out that it would make the model ``too powerful'', as just one bit of global memory would already make it possible to solve e.g.\ leader election (and this observation already shows that  the \onlineLOCAL model is indeed strictly stronger than the \SLOCAL model). In our work we will see that even though \onlineLOCAL can trivially solve e.g.\ leader election thanks to the global memory, it is not that easy to exploit this extra power in the context of LCL problems. Indeed, \onlineLOCAL turns out to be as weak as \SLOCAL when we look at LCL problems in paths, cycles, and rooted trees.

\paragraph{Local computation algorithms.}

We do not discuss local computation algorithms (LCAs)~\cite{rubinfeld11,10.1007/978-3-642-31594-7_55,doi:10.1137/1.9781611973099.89,even2014deterministic,mansour2013local,rubinfeld11} in this work in more detail, but we briefly point out a direct connection between the \onlineLOCAL model and LCAs. It is known that for a broad family of graph problems (that includes LCLs), we can w.l.o.g.\ assume that whenever the adversary queries a node $v$, the LCA will make probes to learn a \emph{connected} subgraph around node $v$ \cite{goos16nonlocal}. For such problems, an \onlineLOCAL algorithm with locality $T$ is at least as strong as an LCA that makes $T$ probes per query: an LCA can learn some \emph{subgraph} of the radius-$T$ neighborhood of $v$ and, depending on the size of the state space, remember some part of that, while in the \onlineLOCAL model we can learn the entire radius-$T$ neighborhood of $v$ and remember all of that.
We will leave a more detailed exploration of the distinction between distance (how far to see) and volume (how much to see), in the spirit of e.g.\ \cite{rosenbaum20lcl-volume,melnyk22mending-volume}, for future work.

\section{Landscape of models}\label{sec:landscape}

As an introduction to the models, we will first check that all relations in \cref{fig:diagram} indeed hold. In each case, we are interested in \emph{asymptotic} equivalence: for example, when we claim that $A \subsetneq B$, the interpretation is that locality $T$ in model $A$ implies locality $O(T)$ in model $B$, but the converse is not true. Note that the relation between the $\onlineLOCAL$ problems and the online graph algorithms has already been discussed in \cref{sec:intro,sec:def-related}.

\paragraph{Inclusions.}

Let us first argue that the subset relations in \cref{fig:diagram} hold. These cases are trivial:
\begin{itemize}
    \item Any \LOCAL algorithm can be simulated in the \SLOCAL model, and any \SLOCAL algorithm can be simulated in the \onlineLOCAL model (this is easiest to see if one interprets \onlineLOCAL as an extension of \SLOCAL with the global memory).
    \item Any $\dynamicLOCAL^\pm$ algorithm can be directly used in the $\dynamicLOCAL$ model (an algorithm that supports both additions and deletions can handle additions).
\end{itemize}
These are a bit more interesting cases:
\begin{itemize}
    \item To simulate a \LOCAL algorithm $A$ in the $\dynamicLOCAL^\pm$ model, we can simply recompute the entire output with $A$ after each change. If the locality of $A$ is $T$, then the output of $A$ only changes within distance $T$ from a point of change.
    \item To simulate a \dynamicLOCAL algorithm $A$ in the \onlineLOCAL model, we proceed as follows. When the adversary reveals a node~$v$, we feed $v$ along with the new nodes in its radius-$O(T)$ neighborhood to $A$ edge by edge. Now there will not be any further changes within distance $T$ from $v$, and hence $A$ will not change the label $L(v)$ of $v$ any more. Hence the \onlineLOCAL algorithm can also label $v$ with $L(v)$.
\end{itemize}

\paragraph{Separations.}

To prove the separations of \cref{fig:diagram}, we make use of the classic distributed graph problem of \emph{$3$-coloring paths}, as well as the following problems that are constructed to highlight the differences between the models:
\begin{itemize}
    \item \emph{Weak reconstruction:} in each connected component $C$ there has to be at least one node $v$ such that its label $L(v)$ is an encoding of a graph isomorphic to $C$.
    \item \emph{Cycle detection:} for each cycle at least one node of the cycle has to correctly report that it is part of a cycle (more precisely, each node that outputs `yes' has to be part of at least one cycle, and for each cycle there has to be at least one node that outputs `yes').
    \item \emph{Component-wise leader election:} in each connected component exactly one node has to be marked as the leader.
    \item \emph{Nested orientation:} find an acyclic orientation of the edges and label each node recursively with its own identifier, the identifiers of its neighbors, and the labels of its in-neighbors (see \cref{def:nested-orientation} for the precise definition).
\end{itemize}
We can prove the bounds shown in \cref{tab:separations} for the locality of these problems in the five models; see \cref{app:separations,sec:incomparability-dynLOCAL-SLOCAL} for the details. Now each separation in \cref{fig:diagram} follows from one of the rows of \cref{tab:separations}.

\begin{table}
\centering
\begin{tabular}{@{}llllll@{}}
\toprule
Problem
& \LOCAL
& \SLOCAL
& dynamic-
& dynamic-
& online-
\\
&
&
& $\LOCAL^\pm$
& \LOCAL
& \LOCAL
\\
\midrule
$3$-coloring paths & $\Omega(\log^* n)$ & $O(1)$ & $O(1)$ & $O(1)$ & $O(1)$ \\
weak reconstruction & $\Omega(n)$ & $\Omega(n)$ & $O(1)$ & $O(1)$ & $O(1)$ \\
cycle detection & $\Omega(n)$ & $\Omega(n)$ & $\Omega(n)$ & $O(1)$ & $O(1)$ \\
component-wise leader election & $\Omega(n)$ & $\Omega(n)$ & $\Omega(n)$ & $\Omega(n)$ & $O(1)$ \\
nested orientation & $\omega(1)$ & $O(1)$ & $\omega(1)$ & $\omega(1)$ & $O(1)$ \\
\bottomrule
\end{tabular}
\caption{Problems that we use to separate the models, and the bounds that we show for their locality.}\label{tab:separations}
\end{table}

\section{3-coloring bipartite graphs in the \onlineLOCAL model}\label{sec:bipartite-separation}

We will next present our Contribution~\hyperref[ssec:contribution3]{3}: we design an algorithm for $3$-coloring bipartite graphs in the \onlineLOCAL model.
This section will also serve as an introduction into the algorithmic techniques that work in this model.
Equipped with this understanding, we will then in \cref{sec:collapse} start to develop more technical tools that we will need for our Contribution~\hyperref[ssec:contribution2]{2}.

The $3$-coloring algorithm shows the advantage of allowing online algorithms to look around: while the best online coloring algorithm on bipartite graphs is $\Theta(\log n)$-competitive, our algorithm in the \onlineLOCAL model achieves a competitive ratio of $1.5$. Moreover, this algorithm provides an exponential separation between the \SLOCAL and \onlineLOCAL models for a natural LCL problem.

Note that an optimal solution would be to color a bipartite graph with $2$ colors. In all models that we consider here, we know it is not possible to solve $2$-coloring with locality $o(n)$, the worst case being a path with $n$ nodes.
We will show that allowing an \onlineLOCAL algorithm to use only one extra color makes it possible to find a valid coloring with locality $O(\log n)$:
\bipartitecoloring*

In \cref{ssec:bip-coloring}, we will introduce the $3$-coloring algorithm, and we will use the special case of grids in order to visualize it.
Later, in \cref{ssec:grid-coloring}, we will discuss the implications of our algorithm for coloring grids.

\subsection{Algorithm for 3-coloring bipartite graphs}\label{ssec:bip-coloring}

\paragraph{Algorithm overview.}

The high-level idea of our \onlineLOCAL algorithm is to color the presented nodes of the graph with $2$ colors until we see two areas where the $2$-colorings are not compatible. In essence, when the adversary presents a node far from any other node we have seen, we blindly start constructing a $2$-coloring. When the adversary presents nodes in the neighborhood of already colored nodes, we simply expand the $2$-colored ``blob''. We keep expanding such properly $2$-colored blobs until, eventually, two blobs with incompatible $2$-colorings meet (i.e., blobs that have different \emph{parities}). Then, we will use the third color in order to create a \emph{barrier} around one of the blobs, effectively flipping its parity. Our algorithm thereby makes use of the knowledge of previously queried neighborhoods that is given to us by the \onlineLOCAL model: the algorithm is \emph{committing} to colors for nodes in the revealed subgraphs before they are queried.

\paragraph{Algorithm in detail.}

\arxivfigurepageon

\begin{figure}
\centerline{\includegraphics[scale=.45]{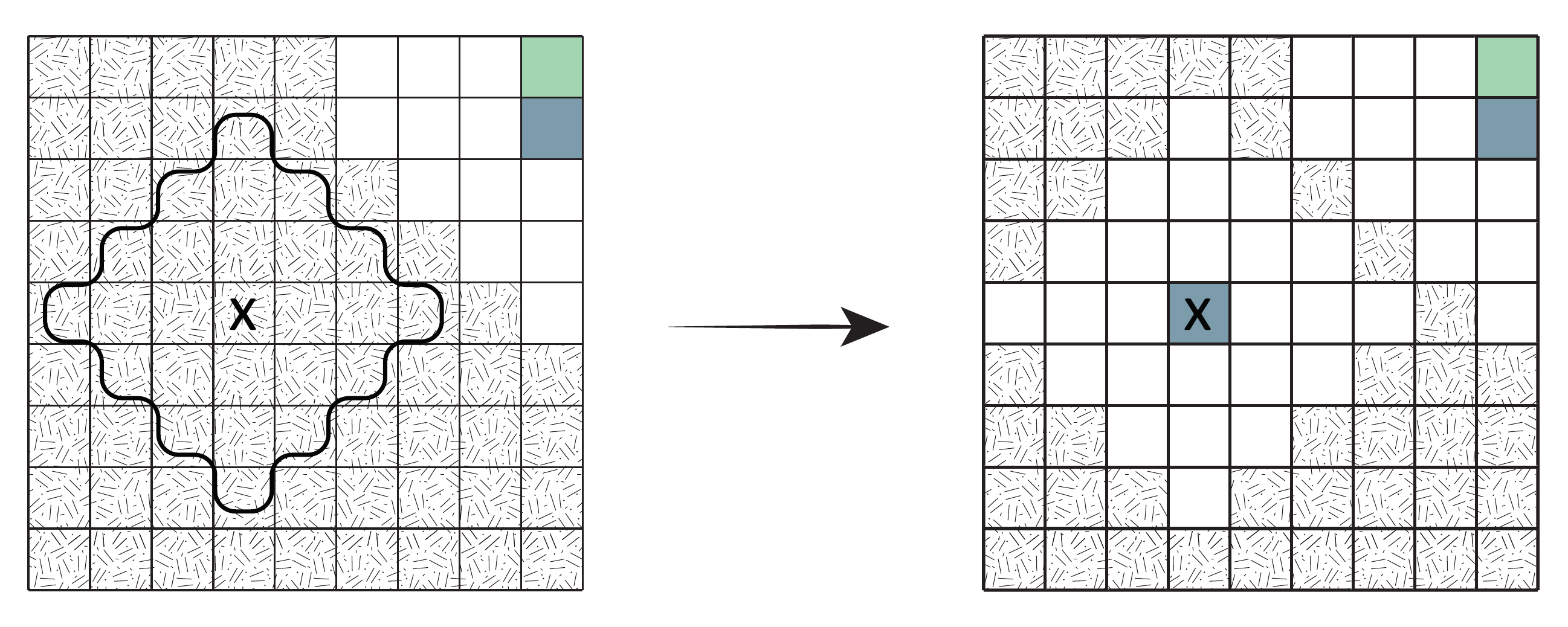}}
\caption{$3$-coloring algorithm, case 1/3. The adversary presents node $x$. Here node $x$ is in the middle of an unseen region (shaded). It will create a new group (white) and we will fix the color of node $x$ arbitrarily.}
\label{fig:3coloring2dgrid-case1}
\Description{}
\end{figure}

\begin{figure}
\centerline{\includegraphics[scale=.45]{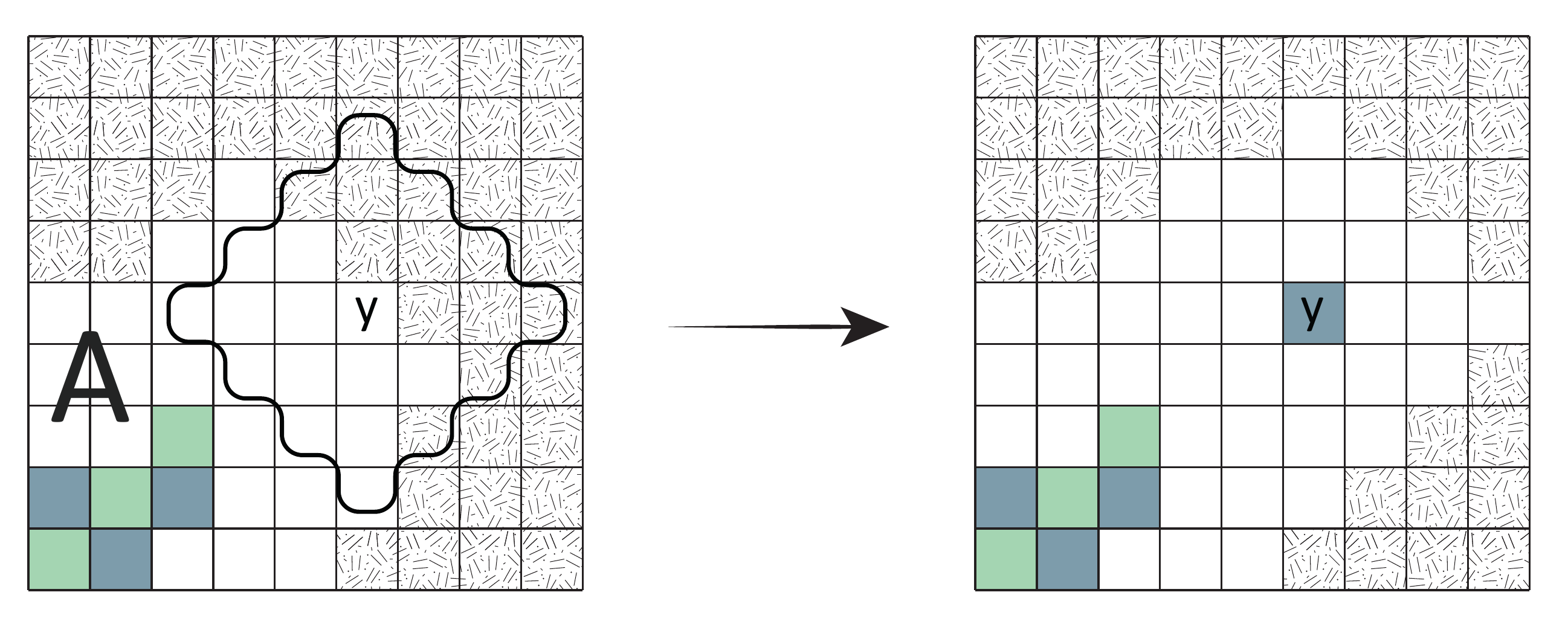}}
\caption{$3$-coloring algorithm, case 2/3. The adversary presents node $y$. Some nodes in the local neighborhood of $y$ are already part of a group (white), and hence $y$ joins this group. We will fix the color of node $y$ so that it is consistent with the coloring of the group.}
\label{fig:3coloring2dgrid-case2}
\Description{}
\end{figure}

\begin{figure}
\centerline{\includegraphics[scale=.45]{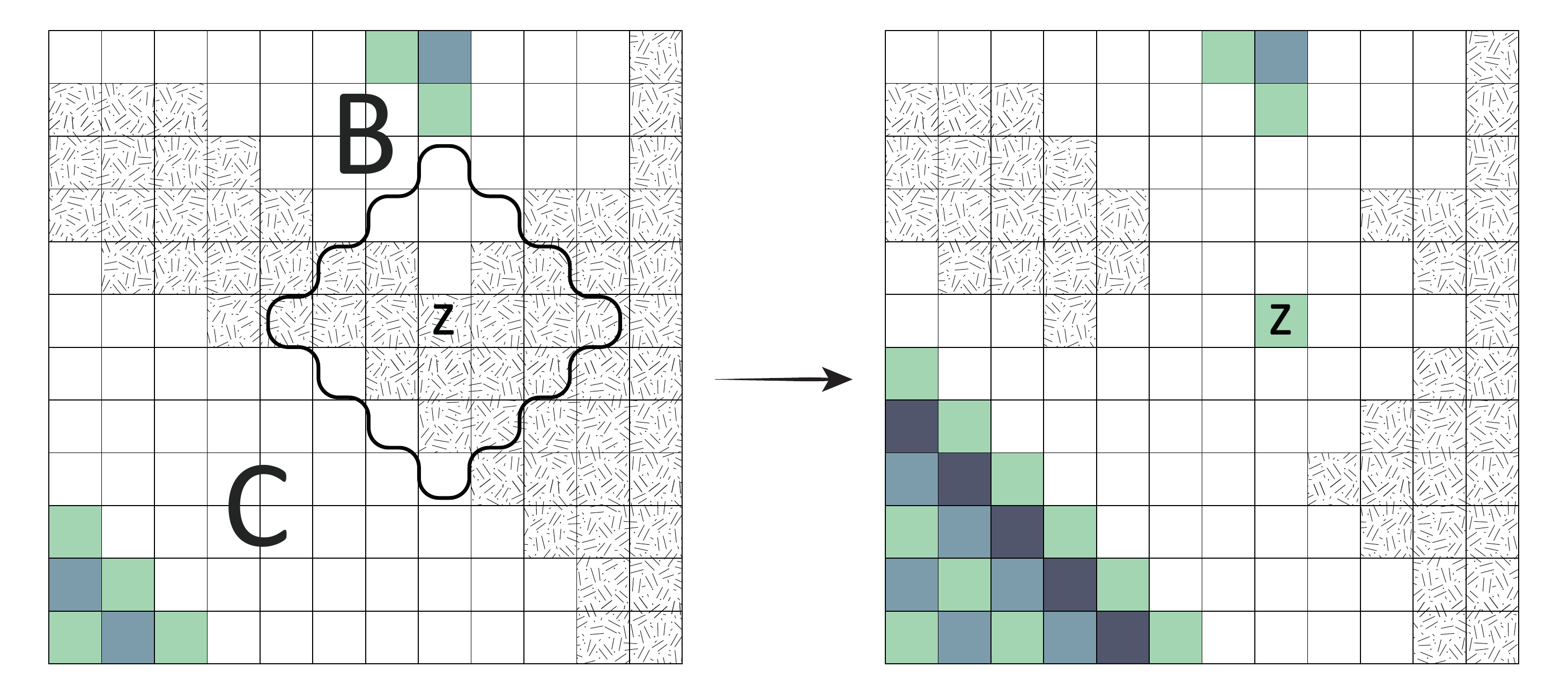}}
\caption{$3$-coloring algorithm, case 3/3. The adversary presents node $z$. Some nodes in the local neighborhood of $z$ belong to two different groups, $B$ and $C$. Hence we merge the groups. As they had incompatible parities, we add a new border around one of the groups, in this case $C$, as both groups had the same number of borders around them and we can choose arbitrarily. Nodes in the local neighborhood of $z$ join the group, and $z$ is colored in a way compatible with the coloring of the newly created group.}
\label{fig:3coloring2dgrid-case3}
\Description{}
\end{figure}

\arxivfigurepageoff

At the beginning, no nodes are revealed to the algorithm, and we therefore say that all nodes are \emph{unseen}. We will refer to connected components of the subgraph $G_i$ of $G$ as \emph{groups}. With each of these groups, we associate a \emph{border count}, which is a natural number that is initially $0$. We will use colors $0$ and $1$ for the $2$-coloring, reserving color $2$ as the barrier color. Each time the adversary points at a node $v_i$, we get to see the radius-$T$ neighborhood $B(v_i, T)$ of this node. Now consider different types of nodes in $B(v_i, T+1)$. There are three different cases that we need to address (we will visualize them in \cref{fig:3coloring2dgrid-case1,fig:3coloring2dgrid-case2,fig:3coloring2dgrid-case3} using grids as an example):
\begin{enumerate}
    \item \emph{All nodes in $B(v_i, T+1)$ are unseen.} In this case, the nodes in $B(v_i, T) $ form a new connected component, i.e. a new \emph{group}. This group has a \emph{border count} of $0$. We will color $v_i$ with $0$, thus fixing the parity for this group (see \cref{fig:3coloring2dgrid-case1}).
    \item \emph{We have already seen some nodes in $B(v_i, T+1)$, but all of them belong to the same group.} In this case, we have been shown an area next to an existing group. If $v_i$ was already committed to a color, we will use that color. Otherwise, we will color $v_i$ according to the $2$-coloring of the group. All nodes in $B(v_i, T)$ are now considered to be in this group (see \cref{fig:3coloring2dgrid-case2}).
    \item \emph{There are nodes in $B(v_i, T+1)$ that belong to different groups.} In this case, we will have to join groups. Here, we will only define the join of two groups $A$ and $B$; if there are more groups, this join can be applied iteratively.
    
    If $A$ and $B$ have different parities (i.e., the $2$-colorings at their boundaries are not compatible), we will take the group with the smaller border count and use a layer of nodes of color $2$ to create a barrier that changes its parity, and then we increase its border count; see \cref{alg:join} for the details. Then, we will join the groups, that are now compatible, and set the border count of the newly created group to the maximum of the border counts of $A$ and~$B$.
    
    By merging all groups in the local neighborhood of $v_i$, we eventually end up in a situation where $v_i$ only sees nodes in a single group, and we are in a scenario similar to case~2 above: nodes in the local neighborhood of $v_i$ also join the newly created group, and if $v_i$ has not already committed to a color, we will color it according to the $2$-coloring of this group (see \cref{fig:3coloring2dgrid-case3}).
\end{enumerate}

\RestyleAlgo{ruled}
\LinesNumbered
\begin{algorithm}[t]
\caption{join\_groups(A, B)}\label{alg:join}
\KwIn{Groups $A, B$}
\KwOut{Group $X$}

\If{A and B have different parities}{
    Let $S$ be the group with the smaller border count. If they are equal, $S=A$\;
    For all nodes of color $0$ in $S$, commit all uncolored neighbors to color $1$\;
    \label{step:join1}
    For all nodes of color $1$ in $S$, commit all uncolored neighbors to color $2$\;
    \label{step:join2}
    For all nodes of color $2$ in $S$, commit all uncolored neighbors to color $0$\;
    \label{step:join3}
    Increase border count of $S$ by $1$.
}
Set all nodes in groups $A, B$ to be in group $X$\;
Set the border count for $X$ to be the maximum of border counts for $A$, $B$ and $S$\;
\Return{$X$}

\end{algorithm}

\subsection{Algorithm analysis}

In order to show correctness of the presented algorithm, we will first prove that this process creates a valid $3$-coloring provided that all our commitments remain within the visible area, that is, inside subgraph $G_i$. Next, we will show that by choosing $T(n)= O(\log n)$, all our commitments will indeed remain inside the visible area. Together, these parts prove \cref{thm:bipartitecoloring}.

\paragraph{Validity of the 3-coloring.}

We will first prove that our algorithm will always continue a valid $3$-coloring, as long as we do not have to make commitments to unseen nodes.
We will consider all three cases of the algorithm individually.
\begin{enumerate}
    \item \emph{All nodes in $B(v_i, T+1)$ are unseen.} In this case, we will color $v_i$ with $0$. As all neighboring nodes were unseen, they have not committed to any color, and thus this case will cause no errors.
    \item \emph{We have already seen some nodes in $B(v_i, T+1)$, but all of them belong to the same group.} In this case, we would either use the committed color or the parity of the group. As previously committed colors do not cause errors, and the group has consistent parity, this case cannot cause any errors.
    \item \emph{There are nodes in $B(v_i, T+1)$ that belong to different groups.} In this case, we want to join groups without breaking the coloring. If the two groups had the same parity, clearly, no errors can be caused by continuing the $2$-coloring. The interesting case is when the two groups had different parities. Then, we need to show that the new commitments made by \cref{alg:join} do not create any errors. 
    
    Let $S$ be the group with the smaller border count. By examining \cref{alg:join}, we can see that all colored nodes that have uncolored neighbors are either of color $0$ or color $1$: only in \cref{step:join2}, nodes can be colored with color $2$, and all of those nodes' neighbors are then colored in \cref{step:join3}. Thus, in order for an error to occur, there either needs to be two nodes of colors $0$ and $1$ that have uncolored neighbors and different parities in $S$, or the algorithm commits to a color of a node that it has not yet seen. This could cause an error, as two groups could commit a single node to two different colors.
    
    As for the first case, we assume that all nodes in $S$ that have uncolored neighbors also have consistent parity. This trivially holds for a group that has border count $0$, as all colored nodes in it have the same parity. From the assumption, it follows that all nodes colored with $1$ in \cref{step:join1} have the same parity, so they cannot create an error. After this, all colored nodes with uncolored neighbors in the group have the same parity, and are colored with $1$. Thus all nodes colored with $2$ in \cref{step:join2} also have the same parity, as do the nodes colored with $0$ in \cref{step:join3}. As these are the only lines where nodes are colored, this procedure cannot create any errors. It also ensures that, after the procedure, the only colored nodes in the group that have uncolored neighbors are the nodes colored in \cref{step:join3}, which have the same parity. Therefore, our assumption holds for all groups. Those nodes also have a parity different from the nodes in $S$ that had uncolored neighbors before this procedure, so in essence, we have flipped the parity of group $S$ to match the parity of the other group.
    
    As for the second case, this can be avoided by choosing a large enough $T$, so that all commitments remain within the visible area of $G_i$. Next, we will discuss how to choose such a~$T$.
\end{enumerate}

\paragraph{Locality of the 3-coloring algorithm.}

In this part, we need to prove that by choosing locality $T(n) = 3\lceil \log_2 n\rceil = O(\log n)$, no nodes outside the visible area of $G_i$ need to be committed.

We first make the observation that a group with border count $b$ contains at least $2^b$ nodes; this is a simple induction:
\begin{enumerate}[leftmargin=10ex]
    \item[$b=0$:] A newly created group contains at least $1$ node.
    \item[$b>0$:] Consider the cases in which \cref{alg:join} returns a group $X$ with border count $b$. One possibility is that $A$ or $B$ already had border count $b$, and hence by assumption it already contained at least $2^b$ nodes. The only other possibility is that both $A$ and $B$ had border count exactly $b-1$, they had different parities, one of the border counts was increased, and hence $X$ has now got a border count of $b$. But, in this case, both $A$ and $B$ contained at least $2^{b-1}$ nodes each.
\end{enumerate}
Hence the border count is bounded by $b \le \log_2 n$ in a graph with $n$ nodes.

We will next consider the maximum distance between a node that the adversary has pointed at and a node with a committed color.
Note that the only place where the algorithm commits a color to a node that the adversary has not revealed yet is when building a border around a group.
There are three steps (\cref{step:join1,step:join2,step:join3}) where the algorithm commits to the color of a neighbor of a committed node, and thus effectively extends the distance by at most one in each step. Therefore, if the border count is $b$, we will in the worst case commit the colors within distance $3b$ from a node that was selected by the adversary.

As we have $b \le \log_2 n$, a locality of $3 \lceil\log_2 n\rceil \ge 3b$ suffices to ensure that all of our commitments are safely within the visible region. This concludes the proof of \cref{thm:bipartitecoloring}.

\subsection{Locality of vertex coloring in grids}
\label{ssec:grid-coloring}

In this section, we give an overview of the locality of coloring in the \LOCAL, \SLOCAL and \onlineLOCAL models in two-dimensional grid graphs with $n = a \times b$ nodes. A two-dimensional grid $G$ is the Cartesian product of a path with $a$ nodes and a path with $b$ nodes.

\paragraph{2-coloring in grids is non-local.}

A grid graph can be colored with $2$ colors, but $2$-coloring is an inherently global problem; in all models that we consider here we know that it is impossible to solve this problem with locality $o(n)$. To see this, observe that we can have a grid with dimensions  $n \times 1$, that is, a path. For a path, it is easy to verify that by coloring just one node, we define the colors of all other nodes along the path. Therefore,  $2$-coloring a path and consequently a grid requires locality $\Omega(n)$. 

\paragraph{5-coloring in grids is trivially local.}

Looking at the other extreme, $5$-coloring is trivial, as one can construct such a coloring greedily (note that in a grid all nodes have degree at most $4$ and can therefore always choose a ``free'' color). The locality of $5$-coloring is therefore $0$ in the \onlineLOCAL, $1$ in the SLOCAL, and $O(\log^* n)$ in the LOCAL model.

\paragraph{4-coloring in grids is also local.}

It turns out that $4$-coloring is also doable with locality $O(\log^* n)$ in the LOCAL model~\cite{brandt17grid-lcl,balliu22mending}. These algorithms can be translated into a constant-locality algorithm in \onlineLOCAL and SLOCAL models.

This problem is also a good illustration of the power of the \onlineLOCAL model. A small constant locality is sufficient for \onlineLOCAL to solve a natural graph problem, while the problem is not solvable in the usual online graph algorithms setting.

\paragraph{3-coloring separates the models.}

So far, we have covered coloring with $2$ colors and coloring with at least $4$ colors, and in each of these cases we have gotten approximately the \emph{same} locality in the \onlineLOCAL and \LOCAL models. However, $3$-coloring turns out to be a much more interesting problem.

By prior work \cite{brandt17grid-lcl}, it is known that the locality of $3$-coloring in $\sqrt{n} \times \sqrt{n}$ grids is at least $\Omega(\sqrt{n})$ in the \LOCAL model. The aforementioned paper considered the case of \emph{toroidal} grid graphs, but the same argument can be applied for non-toroidal grids (in essence, if you could color locally anywhere in the middle of a non-toroidal grid, you could also apply the same algorithm to color a toroidal grid).
We can easily extend this result to show a polynomial lower bound for $3$-coloring grids in the \SLOCAL model:
\begin{restatable}{theorem}{thmSlocalGrids}\label{thm:slocalGrids}
    There is no \SLOCAL algorithm that finds a $3$-coloring in $2$-dimensional grids with locality $o(n^{1/10})$.
\end{restatable}
To prove the result, we show that we can simulate \SLOCAL algorithms sufficiently efficiently in the \LOCAL model. We use the standard technique of first precomputing a distance-$o(n^{1/10})$ coloring, and then using the colors as a schedule for applying the \SLOCAL algorithm. Such a simulation can be done efficiently and would lead to a \LOCAL algorithm running in $o(\sqrt{n})$ time, which is a contradiction.
In \cref*{app:slocal-3-coloring}, we present a full proof of the lower bound.

As grids are bipartite graphs, \cref{thm:bipartitecoloring} gives the following corollary:
\begin{corollary}
    There is an \onlineLOCAL algorithm that finds a $3$-coloring in $2$-dimensional grids with locality $O(\log n)$.
\end{corollary}
\noindent In other words, the problem of $3$-coloring in grids already gives an exponential separation between the \SLOCAL and \onlineLOCAL models. We have summarized the known locality bounds in \cref{tab:coloring}.

\section{LCL problems in paths, cycles, rooted regular trees}\label{sec:collapse}

We just showed that the \onlineLOCAL model is much more powerful than \LOCAL and \SLOCAL for an LCL on bipartite graphs and grids.
In this section, we will discuss what happens when we restrict our attention to LCL problems in paths, cycles, and trees. We will start by defining LCL problems more formally.

We say that $\Pi$ is a \emph{locally verifiable problem} with verification radius $r$ if the following holds: there is a collection of labeled local neighborhoods $\mathcal{T}$ such that $L$ is a feasible solution for input $(G,I)$ if and only if for all nodes $v$, the radius-$r$ neighborhood of $v$ in $(G,I,L)$ is in $\mathcal{T}$. Informally, a solution is feasible if it looks good in all radius-$r$ neighborhoods.

\begin{definition}[Locally checkable labeling \cite{naor95}]\label{def:lcl}
    A locally verifiable problem $\Pi$ is a \emph{locally checkable labeling (LCL)} problem if the set of the input labels $\Sigma$ is finite, the set of the output labels $\Gamma$ is finite, and there is a natural number $\Delta$ such that maximum degree of any graph $G\in\mathcal{G}$ is at most~$\Delta$.
\end{definition}
\noindent Note that in LCL problems, $\mathcal{T}$ is also finite since there are only finitely many possible non-isomorphic labeled local neighborhoods.

It turns out that in the case of paths, cycles, and rooted regular trees, the \LOCAL, \SLOCAL, \dynamicLOCAL, and \onlineLOCAL models are all approximately equally expressive for LCL problems.
In particular, all classification and decidability results related to LCLs in paths, cycles, and rooted regular trees in the \LOCAL model \cite{balliu21rooted-trees,chang21automata-theoretic,balliu19lcl-decidability} directly apply also in the \onlineLOCAL model and both versions of the \dynamicLOCAL model.

We are first going to show that the \LOCAL and \onlineLOCAL models are equivalent in the case of paths and cycles, even when the LCL problems can have inputs.
We will then prove that the models are equivalent also in the more general case of rooted regular trees, but in this case we do not consider the possibility of having input labels.

Formally, we will prove the following theorem for cycles and paths:
\begin{restatable}{theorem}{thmcyclepath}\label{thm:cyclepaths}
    Let $\Pi$ be an LCL problem in paths or cycles (possibly with inputs). If the locality of $\Pi$ is $T$ in the \onlineLOCAL model, then its locality is $O(T + \log^* n)$ in the \LOCAL model.
\end{restatable}
For the case of rooted trees, we will prove the following two theorems:
\begin{theorem}[restate = thmtreesuper, name = ]\label{thm:rooted-tree-superlogarithmic}
    Let $\Pi$ be an LCL problem in rooted regular trees (without inputs).
    Problem $\Pi$ has locality $n^{\Omega(1)}$ in the \LOCAL model if and only if it has locality $n^{\Omega(1)}$ in the \onlineLOCAL model.
\end{theorem}
\begin{theorem}[restate = thmtreesub, name = ]\label{thm:rooted-tree-sublogarithmic}
    Let $\Pi$ be an LCL problem in rooted regular trees (without inputs).
    Problem $\Pi$ has locality $\Omega(\log n)$ in the \LOCAL model if and only if it has locality $\Omega(\log n)$ in the \onlineLOCAL model.
\end{theorem}
These two theorems show that all LCL problems in rooted regular trees belong to one of the known complexity classes $O(\log* n)$, $\Theta(\log n)$ and $n^{\Omega(1)}$ in all of the models we study. In what follows, we will introduce the high-level ideas of the proofs of these theorems.
For full proofs, we refer the reader to \cref{sec:full-proof-paths-cycles,sec:full-proof-rooted-trees}.

\subsection{Cycles and paths}\label{ssec:collapse-paths-cycles}

We will prove \cref{thm:cyclepaths} by first showing that any LCL problem in cycles and paths has either locality $O(1)$ or $\Omega(n)$ in the \onlineLOCAL model.
Next, we will show that if a problem is solvable with locality~$O(1)$ in the \onlineLOCAL model, then it is also solvable in locality~$O(\log* n)$ in the \LOCAL model. These steps are described by the following two lemmas:
\begin{restatable}{lemma}{lemmacyclepathspeedup}\label{lemma:1dspeedup}
    Let $\Pi$ be an LCL problem in paths or cycles (possibly with inputs), and let $\algoa$ be an \onlineLOCAL algorithm solving $\Pi$ with locality $o(n)$. Then, there exists an \onlineLOCAL algorithm $\algoa'$ solving $\Pi$ with locality $O(1)$.
\end{restatable}

The high-level idea of the proof of \cref{lemma:1dspeedup} is to construct a large virtual graph $P'$ such that when the original algorithm runs on the virtual graph $P'$, the labeling produced by the algorithm is locally compatible with the labeling in the original graph $P$.
We ensure this by applying a pumping-lemma-style argument on the LCL problem.
The proof uses similar ideas as the ones presented by \citet*{Chang2017ATH}.
\begin{restatable}{lemma}{lemmacyclepathsimulation}\label{lemma:1dsimulation}
    Let $\Pi$ be an LCL problem in paths or cycles (possibly with inputs), and let $\algoa$ be an \onlineLOCAL algorithm solving $\Pi$ with locality $O(1)$. Then, there exists a \LOCAL algorithm $\algoa'$ solving $\Pi$ with locality $O(\log* n)$.
\end{restatable}

The high-level idea of the proof of \cref{lemma:1dsimulation} is to use the constant locality \onlineLOCAL algorithm to construct a canonical output labeling for each possible neighborhood of input labels.
The fast LOCAL algorithm can then use these canonical labelings in disjoint neighborhoods of the real graph, and the construction of the canonical labelings ensures that the labeling also extends to the path segments between these neighborhoods.

The full proofs of these lemmas can be found in \cref{sec:full-proof-paths-cycles}. In order to prove \cref{thm:cyclepaths}, it is sufficient to combine these lemmas with the fact that the possible localities on paths and cycles in the \LOCAL model are $O(1)$, $\Theta(\log*n)$ and $\Theta(n)$.

\subsection{Rooted regular trees}\label{ssec:collapse-rooted-trees}

We prove the equivalence of the \LOCAL and the \onlineLOCAL models in two parts.
We start out with \cref{thm:rooted-tree-superlogarithmic} and show that if an LCL problem requires locality~$n^{\Omega(1)}$ in the \LOCAL model, then for every locality-$n^{o(1)}$ \onlineLOCAL algorithm we can construct an input instance which the algorithm must fail to solve.
To prove \cref{thm:rooted-tree-sublogarithmic}, we show that a locality-$o(\log n)$ \onlineLOCAL algorithm for solving an LCL problem implies that there exists a locality-$O(\log* n)$ \LOCAL algorithm for solving that same problem. In the following, we will outline the proofs of both theorems, the full proofs can be found in \cref{sec:full-proof-rooted-trees}. Before considering the full proof, we advise the reader to look at the example in~\cref{sec:rooted-tree-example}, where we show that the $\twohalf$-coloring problem requires locality~$\Omega(\sqrt{n})$ in the $\onlineLOCAL$ model.

\begin{proof}[Proof outline of \cref{thm:rooted-tree-superlogarithmic}]\let\qed\relax
Our proof is based on the fact that any LCL problem requiring locality~$n^{\Omega(1)}$ in the \LOCAL model has a specific structure.
In particular, the problem can be decomposed into a sequence of \emph{path-inflexible} labels and the corresponding sequence of more and more restricted problems~\cite{balliu21rooted-trees}.
Informally, a label is path-inflexible if two nodes having that label can exist only at specific distances apart from each other.
For example, when 2-coloring a graph, two nodes having label \texttt{1} can exist only at even distances from each other.
The problems in the path-inflexible decomposition are formed by removing the path-inflexible labels from the previous problem in the sequence until an empty problem is reached.

This decomposition of the problem into restricted problems with path-inflexible labels allows us to construct an input graph for any locality-$n^{o(1)}$ \onlineLOCAL algorithm.
In particular, we force the algorithm to commit labels in disjoint fragments of the graph.
Any label that the algorithm uses must be a path-inflexible label in some problem of the sequence of restricted problems.
By combining two fragments containing labels that are path-inflexible in the same problem, we can ensure that the algorithm cannot solve that problem in the resulting graph.
Hence the algorithm must use a label from a problem earlier in the sequence.
At some point, the algorithm must use labels that are path-inflexible in the original problem.
At that point, we can combine two fragments having path-inflexible labels in the original problem in such a way that no valid labeling for the original problem exists, and hence the algorithm must fail to solve the problem on the resulting graph.
\end{proof}

\begin{proof}[Proof outline of \cref{thm:rooted-tree-sublogarithmic}]\let\qed\relax
Here, we show that a locality-$o(\log n)$ \onlineLOCAL algorithm solving an LCL problem implies that there exists a \emph{certificate for $O(\log* n)$ solvability} for that problem.
It is known that the existence of such a certificate for a problem implies that there exists a locality-$O(\log* n)$ \LOCAL algorithm for solving the problem~\cite{balliu21rooted-trees}.

Informally, the certificate for $O(\log* n)$ solvability for LCL problem $\Pi$ with label set $\Gamma$ and arity $\delta$ consists of a subset $\Gamma_{\mathcal{T}} = \{ \gamma_1, \ldots, \gamma_t \}$ of labels $\Gamma$, and two sequences of correctly labeled complete $\delta$-ary trees $\mathcal{T}^1$ and $\mathcal{T}^2$.
The leaves of each tree in the sequence $\mathcal{T}^1$ (resp. $\mathcal{T}^2$) are labeled in the same way using only labels from set $\Gamma_{\mathcal{T}}$.
For every label of set $\Gamma_\mathcal{T}$, there exists a tree in both of the sequences having a root labeled with that label.

We can use the \onlineLOCAL algorithm to construct such a certificate.
We do this by constructing exponentially many deep complete $\delta$-ary trees and using the algorithm to label nodes in the middle of those trees.
We then glue these trees together in various ways.
When the trees are glued together, we use the \onlineLOCAL algorithm to label the rest of the nodes to form one tree of the sequence.
We repeat this procedure until all trees of both sequences have been constructed.
\end{proof}

\begin{acks}
We would like to thank Alkida Balliu, Sameep Dahal, Chetan Gupta, Fabian Kuhn, Dennis Olivetti, Jan Studený, and Jara Uitto for useful discussions. We would also like to thank the anonymous reviewers for the very helpful feedback they have provided for previous versions of this work. This work was supported in part by the Academy of Finland, Grant 333837.
\end{acks}

\bibliography{references}

\newpage

\appendix

\section{Simple separations}\label{app:separations}

We show the upper and lower bounds for the first four rows of \cref{tab:separations} here; the final row is left for \cref{sec:incomparability-dynLOCAL-SLOCAL}.

\paragraph{3-coloring paths.}
This problem is easy to solve greedily in \SLOCAL and maintain greedily in $\dynamicLOCAL^\pm$ with locality $O(1)$.
However, the locality of this problem in the \LOCAL model is known to be $\Omega(\log^* n)$ \cite{linial92}.

\paragraph{Weak reconstruction.} Recall that in this problem at least one node has to correctly output the structure of its own connected component.

In the $\dynamicLOCAL^\pm$ model we can solve it with locality $O(1)$: after each change, let the nodes located next to the point of change report the current structure of their own connected components.

In the $\SLOCAL$ model this cannot be solved with locality $o(n)$. To see this, assume that an $o(n)$-locality \SLOCAL algorithm $\algoa$ exists, and fix a sufficiently large $n$. Let $\mathcal{S} = \{S_1,S_2,\dotsc,S_k\}$ be a family of non-isomorphic graphs with $n/6$ nodes each; we assume that in each graph $S_i$ there is a unique \emph{root node} that we can distinguish (e.g., the only node with degree $2$). We can choose $k \gg n$.
\begin{center}
\includegraphics[page=1]{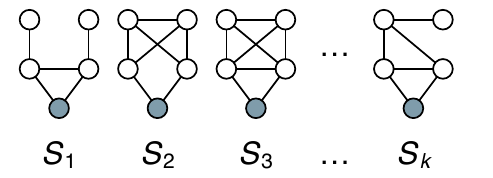}
\end{center}
Now for each $1 \le i < j \le k$ let $G_{i,j}$ be a graph that consists of $S_i$ and $S_j$ such that their root nodes are connected with a path $P$ of length $2n/3$. We split $P$ in three parts: \emph{head} ($n/6$ nodes closest to $S_i$), \emph{tail} ($n/6$ nodes closest to $S_j$), and \emph{middle} (the remaining $n/3$ nodes far from $S_i$ and $S_j$).
\begin{center}
\includegraphics[page=2]{figures/figs.pdf}
\end{center}
Now we consider what happens when we present the middle part to $\algoa$ first, then the head and $S_i$, and finally the tail and $S_j$. Each node $v$ outputs some label $L(v)$; if $\algoa$ works correctly, one of the labels encodes a graph isomorphic to $G_{i,j}$. Let $M$ be the set of labels produced by the middle part, let $H_i$ be the set of labels produced by the nodes in the head and $S_i$, and let $T_j$ be the set of labels produced by the nodes in the tail and $S_j$. The key observation is that set $M$ is independent of $i$ and $j$, set $H_i$ only depends on $i$, and set $T_j$ only depends on $j$. This is because $\algoa$ produces the same output if we present the middle part first, then the tail and $S_j$, and only after that the head and $S_i$. Furthermore, we know that $|M| \le n/3$, $|H_i| \le n/3$, and $|T_j| \le n/3$.

Let $X = M \cup \bigcup_i H_i \cup \bigcup_j T_j$ be the union of all labels that $\algoa$ may produce in the graph family $G_{i,j}$ for the processing order described above. We have $|X| \le (1 + k + k)\cdot n/3 \ll k^2$. As there are $k^2$ nonisomorphic graphs $G_{i,j}$, there has to exist a pair $(i,j)$ such that no label in $X$ is an encoding of $G_{i,j}$. Hence if we apply $\algoa$ to $G_{i,j}$ with the above processing order, none of the nodes will correctly output the structure of $G_{i,j}$.

\paragraph{Cycle detection.} Recall that in this problem for each cycle at least one node of the cycle has to correctly report that it is part of a cycle.

In the \dynamicLOCAL we can solve it with constant locality: the first node that encounters a closed cycle reports it.

In the $\dynamicLOCAL^\pm$ model the problem cannot be solved with locality~$o(n)$. To see this, first construct a long cycle; at least one node $v$ will report that it is part of a cycle. Then break the cycle by deleting an edge that is far from $v$; node $v$ cannot respond to this change, and will hence still incorrectly report that it is part of a cycle.

In the $\SLOCAL$ model the problem cannot be solved with locality~$o(n)$, either. To see this, first consider a long cycle $C$, and split the cycle in four parts of the same length: $C_1, C_2, C_3, C_4$ (numbered in the positive direction along the cycle). Consider the following two processing orders:
(i) $C_1, C_3, C_2, C_4$ and
(ii) $C_1, C_3, C_4, C_2$.
Note that both orders will result in exactly the same output. Furthermore, at least one node has to report that it is part of a cycle. Thus we can always find a processing order in which a node in one of the first three parts reports that it is part of a cycle. But then we can leave out an edge in the fourth part to derive a contradiction (the algorithm claims that there is a cycle while we have a path).

\paragraph{Component-wise leader election.} Recall that in this problem in each connected component exactly one node has to be marked as the leader.

In the \onlineLOCAL model we can solve this problem with constant locality: the algorithm can see the last node of a connected component and assign this node to be a leader.

In the \dynamicLOCAL model this problem cannot be solved with locality~$o(n)$: start with two components with large diameters, both of which have leaders, and then connect the components with an edge that is far from the leaders. At least one of the leaders would need to be removed, but both of them are far from the point of change.

In the \SLOCAL model this problem cannot be solved with locality~$o(n)$, either. To see this, consider two long paths, $P$ and $Q$. Split $P$ in three parts of the same length: $P_1, P_2, P_3$ (numbered in this order along the path). Consider these two processing orders:
(i) $P_2, P_1, P_3$ and
(ii) $P_2, P_3, P_1$.
Note that both orders will result in exactly the same output. Furthermore, exactly one node has to be elected as the leader. Hence in one of these processing orders, the algorithm will commit to a leader that is located within the first two parts, before seeing the final part. Similarly in $Q$ we can find a processing order in which the algorithm commits to a leader early. Therefore we can add an edge between $P$ and $Q$ in parts that the algorithm has not seen yet. We have a contradiction: the algorithm marked are at least two nodes as leaders in one component.

\section{Nested orientation}
\label{sec:incomparability-dynLOCAL-SLOCAL}

In this section, we show that the \dynamicLOCAL model cannot simulate the \SLOCAL model.
We do this by defining \emph{nested orientation} and showing that finding it is trivial in the \SLOCAL model with constant locality, but impossible in the \dynamicLOCAL model with constant locality.

We start by defining the nested orientation:
\begin{definition}[Nested orientation]\label{def:nested-orientation}
    Let $G = (V, E)$ be a simple undirected graph, and let $f \colon V \to \{1, 2, \ldots, \poly(|V|)\}$ be a mapping of the nodes to unique identifiers.
    A \emph{nested orientation} $(O, h)$ of $G$ consists of an \emph{acyclic orientation}~$O$ of the edges~$E$ and a \emph{nesting}~$h$ of nodes~$V$.
    For a node $v \in V$, we set~$h(v) = (f(v), F(v), H(v))$, where $f(v)$ is the unique identifier of $v$,
    \[
        F(v) = \bigl\{ f(u) : \{v,u\} \in E \bigr\}
    \]
    is the set of unique identifiers of all neighbors of $v$, and
    \[
        H(v) = \bigl\{ h(u) : (u,v) \in O \bigr\}
    \]
    is the set of outputs of all in-neighbors of $v$ according to orientation~$O$.
\end{definition}
Note that a nested orientation exists for any graph $G$, and the orientation can be trivially found in the \SLOCAL model with locality $T=1$:
When a node~$v$ is processed, it is sufficient to orient all undirected adjacent edges away from $v$, and to inspect the identifiers and outputs of all neighbors of $v$ to compute $F(v)$ and $H(v)$, respectively.

We start by proving that the length of the nestings of a constant-locality \dynamicLOCAL algorithm is bounded.
More formally, we prove the following lemma:
\begin{lemma}
    \label{lemma:nestedupperboundpath}
    Let $\algoa$ be a \dynamicLOCAL algorithm that finds a nested orientation with constant locality $T$, and let $G$ be a graph with a girth of at least $2T+2$.
    Let $O$ be the orientation produced by $\algoa$ when run on graph~$G$.
    Then orientation~$O$ does not induce a directed walk of length $T+1$ in $G$.
\end{lemma}
\begin{proof}
    Assume for contradiction that orientation~$O$ induces a walk~$P = v_1 v_2 \dots v_{T+1}v_{T+2}$ of length $T+1$ in graph~$G$.
    In particular, this implies that the nesting $h(v_{T+2})$ is dependent on the output of~$v_1$.

    We can modify the graph $G$ by introducing a new node $x$ and connecting it to $v_1$ by an edge. We then show this change to algorithm~$\algoa$.
    This means that the algorithm needs to change the output of node~$v_1$ to include the identifier of its new neighbor.
    This implies that also the output of $v_{T+2}$ needs to be changed as it is dependent on the output of $v_1$.
    Since the girth of $G$ is at least $2T+2$ (path~$P$ is the shortest path between $v_1$ and $v_{T+2}$), node~$v_{T+2}$ does not belong to the radius-$T$ neighborhood of the change. This is a contradiction.
\end{proof}

We are now ready to show that finding a nested orientation is not possible with constant locality in the \dynamicLOCAL model, and therefore it separates the \SLOCAL model from the \dynamicLOCAL model:
\begin{theorem}\label{thm:nested-orientation-dynlocal-unsolvablity}
    Finding a nested orientation in the \dynamicLOCAL model requires locality~$\omega(1)$.
\end{theorem}
\begin{proof}
    Assume for contradiction that there exists a \dynamicLOCAL algorithm~$\algoa$ that finds a nested orientation with locality $T = O(1)$.
    Let $G = (V, E)$ be a graph with chromatic number and girth of at least $2T$; such graphs are know to exist~\cite{bollobas2004extremal}.
    We will show that running algorithm~$\algoa$ on $G$ produces a coloring with much fewer colors than $2T$, which would contradict our assumption.

    We start by executing algorithm~$\algoa$ on graph~$G$ in an arbitrary order to produce an orientation~$O$ and a nesting~$h$.
    Let $D$ be the directed graph counterpart of $G$ where all edges have been oriented according to $O$; by definition of $O$, graph~$D$ is acyclic.
    This implies that some nodes of $D$ are sinks with only incoming edges.
    We will use these sinks to recursively partition the nodes of $D$, and therefore also $G$, into $L$ independent sets~$C_1, C_2, \ldots, C_L$ as follows:
    \begin{align*}
        D_1 &= D, \\
        C_i &= \text{sink nodes of $D_i$} & \forall i \ge [1, L], \\
        D_{i+1} & = D_i[V_{D_i} \setminus C_i] & \forall i \ge [1, L],
    \end{align*}
    where $D_{L+1}$ is the empty graph.
    See \cref{fig:nested-orientation} for an illustration of this process.
    Note that each set $C_i$ is indeed an independent set of nodes. This is because two adjacent nodes cannot simultaneously be sinks since the edge between them is oriented towards one of the nodes.
    Note also that the longest walk in $D_{i+1}$ is one edge shorter than in $D_i$.
    This, together with \cref{lemma:nestedupperboundpath}, implies that $L \le T$.
    This is however a contradiction as the independent sets~$C_1, C_2, \ldots, C_L$ induce a coloring with $L \le T$ colors for $G$, but the chromatic number of $G$ is at least $2T$ by assumption.
\end{proof}

\begin{figure}
    \centering
    \includegraphics[scale=0.3]{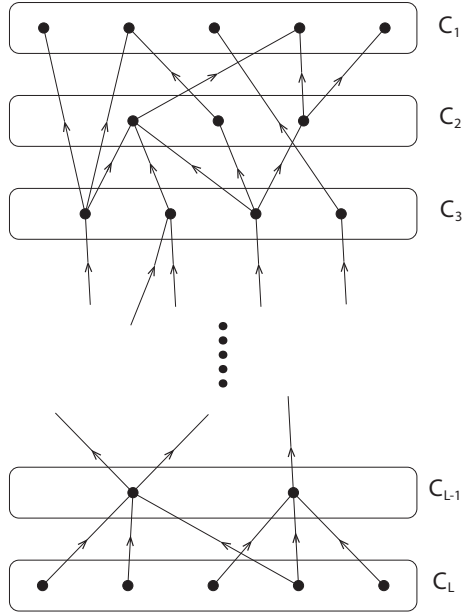}
    \caption{Visualization of partitioning the nodes of $D$ into $L$ independent sets in the proof of \cref{thm:nested-orientation-dynlocal-unsolvablity}. The arrows represent orientation~$O$ of nodes. Note that after we remove nodes of set~$C_1$, that is the sinks in the original directed graph~$D$, then the nodes in set~$C_2$ become sinks, and so on, until $C_L$ contains all the rest of the nodes.}
    \label{fig:nested-orientation}
\end{figure}

\section{Hardness of 3-coloring in SLOCAL}\label{app:slocal-3-coloring}

In this section, we prove the following theorem:
\thmSlocalGrids*
\begin{proof}
Suppose there is an \SLOCAL algorithm $\algoa$ with locality $T(n) = o(n^{1/10})$. We design a \LOCAL algorithm $A'$ that uses $A$ to solve the same problem with locality $T'(n) = o(n^{1/2})$, which is known to be impossible \cite{brandt17grid-lcl}.

Algorithm $\algoa'$ proceeds as follows. Let $k > 2T(n)+1$. Let $G'$ be the $k$th power of the input graph $G$, i.e., nodes within distance $k$ in the input graph are neighbors in $G'$. Now the maximum degree of $G'$ is $d = O(k^2)$. Let $c = d^2$. We then use the graph coloring algorithm by \citet{linial92} to find a proper $c$-coloring of $G'$; this results in a distance-$k$ coloring of $G$. Then we simply proceed by color classes: we first apply $\algoa$ in parallel to all nodes of color $1$, then to all nodes of color $2$, etc. As nodes in one color class are sufficiently far from each other, from the perspective of the output of $\algoa$ it does not matter in which order we process these nodes.

This way $\algoa'$ has simulated $\algoa$. The locality that we need for the coloring phase is $O(d \log^* n) = o(n^{1/2})$, and the locality of the final phase is $O(dc) = o(n^{1/2})$.
\end{proof}

\begin{corollary}
There is no \SLOCAL algorithm that finds a $3$-coloring in bipartite graphs with locality $o(n^{1/10})$.
\end{corollary}

\section{LCL problems in paths and cycles}
\label{sec:full-proof-paths-cycles}

In this section, we prove the following theorem:
\thmcyclepath*
We will prove this theorem by first showing that any LCL problem in cycles and paths has either locality $O(1)$ or $\Omega(n)$ in the \onlineLOCAL model.
Next, we will show that problems that require locality~$\Omega(n)$ in the \onlineLOCAL model are exactly the same problems that require locality~$\Omega(n)$ in the \LOCAL model.
In particular, we will be proving the following two lemmas:
\lemmacyclepathspeedup*
\lemmacyclepathsimulation*
In order to prove \cref{thm:cyclepaths}, it is sufficient to combine these lemmas with the fact that the possible localities on paths and cycles in the \LOCAL model are $O(1)$, $\Theta(\log*n)$ and $\Theta(n)$.

\subsection{From sublinear-locality to constant-locality in online-LOCAL}

The high-level idea of the proof of \cref{lemma:1dspeedup} is to construct a large virtual graph $P'$ such that when the original algorithm runs on the virtual graph $P'$, the labeling produced by the algorithm is locally compatible with the labeling in the original graph $P$.
We ensure this by applying a pumping-lemma-style argument on the LCL problem.
The proof uses similar ideas as the ones presented by \citet*{Chang2017ATH}.

We will start by defining a tripartition of vertices of a path, and use it to define an equivalence relation for path segments. Recall that we use $r$ to denote the verification radius of the LCL problem~$\Pi$.
\begin{definition}[tripartition~\cite{Chang2017ATH}]
    Let $S$ be a path segment, and let $s$ and $t$ be the endpoints of $S$.
    A tripartition of vertices $\xi(S, \{s, t\}) = (D_1, D_2, D_3)$ is defined as
    \[
        D_1 = B(s, r-1) \cup B(t, r-1), \quad
        D_2 = \bigcup_{v \in D_1} B(v, r) \setminus D_1, \quad
        D_3 = V(S) \setminus (D_1 \cup D_2).
    \]
\end{definition}
\begin{definition}[equivalence]
    Let $S$ and $S'$ be path segments, and let $s$, $t$, $s'$ and $t'$ be the endpoints of the segments, respectively.
    Let $\xi(S, \{s, t\}) = (D_1, D_2, D_3)$ and $\xi(S', \{s', t'\}) = (D_1', D_2', D_3')$.
    We define that $S$ and $S'$ are equivalent, denoted by $S \pequiv S'$, if the following two statements hold:
    \begin{enumerate}
        \item The graphs induced by $D_1 \cup D_2$ and $D_1' \cup D_2'$ are isomorphic, including the input.
        Moreover, the isomorphism maps $s$ and $t$ in $S$ to $s'$ and $t'$ in $S'$, respectively.
        \item Let $L$ be an output labeling on the graph induced by $D_1 \cup D_2$, and let $L'$ be the corresponding output labeling on the graph induced by $D_1' \cup D_2'$.
        Then $L$ is extendable to a valid labeling for the whole segment $S$ exactly when $L'$ is extendable to a valid labeling for the whole segment~$S'$.
    \end{enumerate}
    We write $\class(S)$ for the equivalence class containing all path segments $S'$ such that $S' \pequiv S$.
\end{definition}
Observe that the equivalence class of a path segment $S$ depends only on the input labeling in the radius-$(2r-1)$ neighborhood of its endpoints, and on the set of possible output label combinations in the same regions.
As there are only finitely many different input and output labelings for the radius-$(2r-1)$ neighborhood, there are only finitely many different classes of path segments.
In particular, there are only finitely many classes which contain only finitely many elements.
Hence there exists a natural number $\alpha$ such that for every path segment $S$ with $|S| \ge \alpha$, the class $\class(S)$ is infinite.
This implies that for every sufficiently long path segment $S$, there exists another \emph{longer} path segment $S'$ that is equivalent to $S$ with respect to the relation $\pequiv$.

One important property of the class of a path segment is that we can replace any path segment with another segment from the same class of path segments without affecting the labeling of the rest of the path.
The following lemma formalizes this idea:
\begin{lemma}\label{lemma:1dreplace}
    Let $P$ be a path, possibly with inputs, let $S$ be a path segment of $P$, and let $L$ be a valid output labeling of $P$.
    Let $S'$ be another path segment such that $\class(S) = \class(S')$, i.e. $S \pequiv S'$, and let $P'$ be a copy of $P$, where segment $S$ is replaced with the segment $S'$.
    Then, there exists a valid labeling $L'$ of $P'$ such that $L(v) = L'(v')$ for each pair of corresponding nodes $v \in V(P) \setminus V(S)$, $v' \in V(P') \setminus V(S')$.
\end{lemma}
\begin{proof}
    Let $s$ and $t$ be the endpoints of $S$, and let $s'$ and $t'$ be the endpoints of $S'$.
    Let $\xi(S, \{s, t\}) = (D_1, D_2, D_3)$ and $\xi(S', \{s', t'\}) = (D_1', D_2', D_3')$ be tripartitions of $S$ and $S'$, respectively.
    We can now construct $L'$ as follows:
    For each $v \in V(P) \setminus D_3$ and the corresponding $v' \in V(P') \setminus D_3'$, set $L'(v') = L(v)$.
    Since $S \pequiv S'$, there must exist a valid labeling for the nodes in $D_3'$ since $L$ is valid for $V(S) \supseteq D_3$.

    The labeling $L'$ must be locally valid for every $v' \in V(P') \setminus (D_2' \cup D_3')$ as $B(v', r) \subseteq V(P') \cup D_1' \cup D_2'$, and therefore the neighborhood looks identical to that of the corresponding node in $P$.
    By definition of $S \pequiv S'$ and the fact that $L$ is valid for the whole path segment $S$, the labeling $L'$ must also be valid for every $v' \in D_2' \cup D_3'$.
\end{proof}

To simplify the description of the algorithm, we make use of the following observation:
\begin{observation}[composition]\label{obs:composition}
\OnlineLOCAL algorithms $\algoa$ and $\algob$ with localities $T_1$ and $T_2$ can be composed into an \onlineLOCAL algorithm $\algob \circ \algoa$ with locality $T_1 + T_2$.
\end{observation}

We will also make use of the ruling sets:
\begin{definition}[ruling sets]
An $(\alpha, \beta)$-ruling set $R$ of a graph $G$ is a subset of $V(G)$ such that for every $v_i, v_j \in R, v_i \ne v_j$ it holds that $\dist(v_i, v_j) \ge \alpha$, and for every $v \in V(G) \setminus R$ there exists at least one node $u \in R$ such that $\dist(v, u) \le \beta$.
\end{definition}

We are now ready to prove \cref{lemma:1dspeedup}:
\begin{proof}[Proof of \cref{lemma:1dspeedup}]
    Let $\Pi$ be an LCL problem in paths or cycles with checking-radius $r$, let $P$ be a path or a cycle with $n$ nodes, and let $\algoa$ be an \onlineLOCAL algorithm solving $\Pi$ with locality $T(n) = o(n)$.
    Let $\alpha$ be a natural number larger than the length of the longest path belonging to a finite class.
    We can now construct $\algoa'$ solving $\Pi$ with locality $6\alpha$.
    
    Algorithm $\algoa'$ proceeds in three phases (which we can compose using \cref{obs:composition}):
    \begin{enumerate}
        \item In the first phase, the algorithm constructs a $(\alpha, \alpha)$-ruling set $R$ of path $P$.
        This can be done with locality $\alpha$ using the following simple algorithm:
        Whenever the adversary points at a node $v$, check whether any node in $B(v, \alpha)$ already belongs to $R$.
        If not, append $v$ to $R$. Otherwise, $v$ is already ruled by another node of $R$.
        \item\label{step:1dspeedupextend} In the second phase, the algorithm constructs a larger virtual path $P'$ with $N$ nodes and simulates algorithm $\algoa$ on $P'$ in specific neighborhoods.
        Let $N$ be such that
        \begin{equation*}
            \frac{2n}{\alpha}T(N) \ll N .
        \end{equation*}
        Such an $N$ always exists as $T(N) = o(N)$.
        The construction of $P'$ ensures that the labeling produced on the specific neighborhoods of $P'$ are also valid in the corresponding neighborhoods of $P$ in such a way that the labeling is extendable for the whole path $P$.
        Moreover, the construction ensures that the simulation can be done with constant locality.
        
        Virtual path $P'$ can be constructed with the help of the $(\alpha, \alpha)$-ruling set $R$ with locality $3\alpha$ as follows:
        Whenever the adversary points at a node $v$, check whether it belongs to the radius-$2r$ neighborhood of a node $u$ in $R$.
        If this is the case, find at most two other nodes $u_1$ and $u_2$ belonging to $R$ in the radius-$2\alpha$ neighborhood of $u$.
        The only case where there can be less than two nodes in this neighborhood is when the neighborhood contains an endpoint of path $P$.
        In that case, let the corresponding node $u_i$ be the endpoint.
        Let $S_1$ be the path segment between $u$ and $u_1$, and let $S_2$ be the path segment between $u$ and $u_2$.
        
        By the construction of $R$, the lengths of $S_1$ and $S_2$ are in the range $[\alpha, 2\alpha]$, unless an endpoint is reached.
        For each $S_i$, such that $|S_i| \ge \alpha$, we know that $\class(S_i)$ is infinite.
        Hence there exists some $S_i'$ of length at least $N/\alpha$ that satisfies $\class(S_i) = \class(S_i')$.
        Otherwise we set $S_i' = S_i$.
        We construct $P'$ by adding node $u$ and the connections between $u$ and $S_1'$ and $S_2'$ if such connections exist.
        Let $v'$ be the node corresponding to node $v$ in $P'$.
        Algorithm $\algoa'$ uses $\algoa$ to find the label for $v'$ on $P'$ and then uses the same label for $v$.
        By the construction of $P'$, algorithm $\algoa$ never sees further than segments $S_1'$ and $S_2'$.
        Since the labeling is locally consistent on $P'$, the produced labeling is also locally consistent on $P$.
        \cref{fig:1dpumping} visualizes this construction.
        
        \begin{figure}
            \centering
            \includegraphics[scale=0.5]{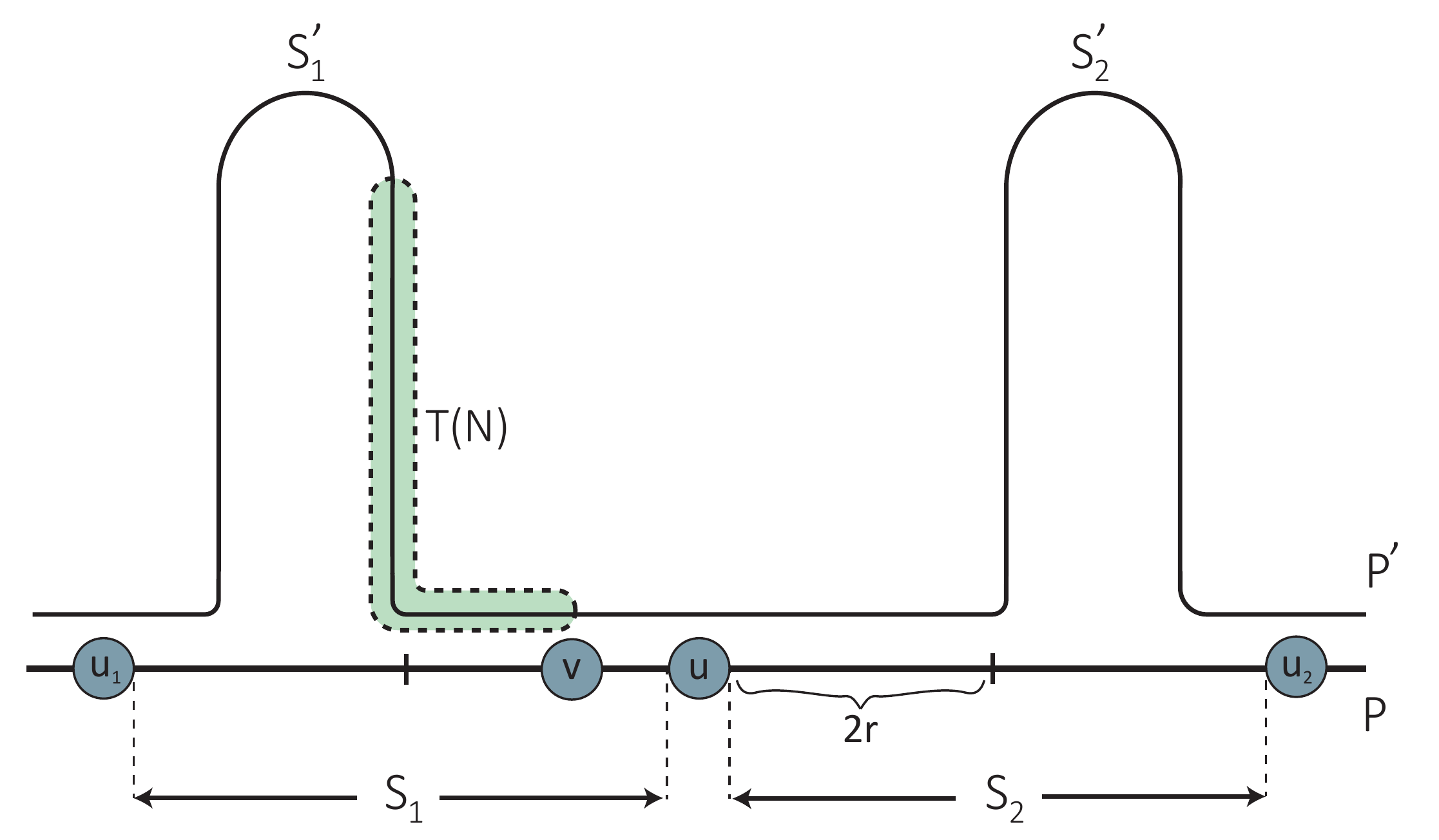}
            \caption{
            Visualization of step~\ref{step:1dspeedupextend} in the proof of \cref{lemma:1dspeedup}.
            The adversary has pointed at node $v$, which is in the radius-$2r$ neighborhood of a node $u$ in the ruling set $R$.
            Algorithm $\algoa'$ finds the neighboring nodes $u_1$ and $u_2$ in $R$, and the corresponding path segments $S_1$ and $S_2$.
            The algorithm then replaces the segments with longer segments $S_1'$ and $S_2'$ from the same class of segments to construct a larger graph $P'$.
            Finally, algorithm $\algoa'$ runs the original algorithm $\algoa$ on the extended path $P'$ to get the label for node $v$.
            Note that by the construction of $P'$, algorithm $\algoa$ does not see further than half of $S_1'$ and $S_2'$.
            }
            \label{fig:1dpumping}
            \Description{}
        \end{figure}
        
        \item\label{itm:1dextend} With the labels around the nodes of $R$ fixed, algorithm $\algoa'$ needs to fix the labeling for the rest of the nodes of $P$.
        By the assumption that $\algoa$ solves the LCL problem $\Pi$, the partial labeling in the radius-$(2r-1)$ neighborhoods of the endpoints of $S_i'$ in the previous step is extendable to the whole path segment $S_i'$.
        \cref{lemma:1dreplace} implies that the same labeling near endpoints is also extendable for the whole path segment $S_i$.
        For each node $v$ that is not inside the radius-$2r$ neighborhood of a node in $R$, algorithm $\algoa'$ finds the corresponding path segment $S_i$ with locality $2\alpha$ and finds a valid label by brute force.
        \qedhere
    \end{enumerate}
\end{proof}

\subsection{From constant-locality in online-LOCAL to log-star-locality in \LOCAL}

Now that we have shown that sublinear-locality \onlineLOCAL algorithms can be sped up to constant-locality algorithms, it remains to prove that constant locality \onlineLOCAL algorithms can be used to construct efficient \LOCAL algorithms.
The high-level idea of the proof of \cref{lemma:1dsimulation} is to use the constant locality \onlineLOCAL algorithm to construct a canonical output labeling for each possible neighborhood of input labels.
The fast \LOCAL algorithm can then use these canonical labelings in disjoint neighborhoods of the real graph, and the construction of the canonical labelings ensures that the labeling extends also to the path segments between these neighborhoods.

\begin{proof}[Proof of \cref{lemma:1dsimulation}]
    Let $\Pi$ be an LCL problem with constant checking-radius $r$, and let $\algoa$ be an \onlineLOCAL algorithm solving $\Pi$ with constant locality $T$.
    Let $\beta = T + r + 1$.

    We first construct an input-labeled path $P$ that, informally, consists of a sufficiently many copies of all possible input neighborhoods. At this point, $P$ is merely a \emph{collection of disjoint path fragments}; we will only later decide how the fragments are connected to each other to form a complete path.
    
    To construct $P$, we first fix a sufficiently large set of \emph{center nodes} $\nu = \{v_1, v_2, \ldots, v_k\}$. Each center node is in the middle of its own path fragment of length $2\beta + 1$. Hence when we finally put the path fragments together, we will satisfy $B(v_i, \beta) \cap B(v_j, \beta) = \emptyset$ for each $i \ne j$.
    
    We choose the input labels of $P$ such that for every possible radius-$\beta$ neighborhood of inputs~$\iota$, there exists some subset of center nodes $M \subseteq \nu$ for which the following holds: $|M| = |\Gamma|^{2r+1}+1$ and $I(B(v, \beta)) = \iota$ for each $v \in M$.
    As there are only finitely many distinct input neighborhoods $\iota$, the size of $P$ can be bounded by a constant.
    
    We will next execute $\algoa$ on every node in the radius-$r$ neighborhood of the center nodes $\nu$ in some arbitrary order and then stop. Note that at this point $P$ is under-specified, but nevertheless, the execution of $\algoa$ at the center nodes is well-defined, as we have already fixed the structure of $P$ in the local neighborhoods of the center nodes. The key idea is that we force $\algoa$ to \emph{commit early} to some outputs around the center nodes, and if $\algoa$ indeed works correctly for any path, then these outputs must be compatible with \emph{any} concrete realization of $P$, no matter how we connect the center nodes to each other. In particular, $\algoa$ must be able to fill the gaps between the center nodes.
    
    We define function $f \colon \Sigma_\Pi^{2\beta+1} \rightarrow \Gamma_\Pi^{2r+1}$ to be a mapping from radius-$\beta$ neighborhoods of input labels to radius-$r$ neighborhoods of output labels such that there are at least two center nodes $v_i$ and $v_j$ in $P$ with $f(B(v_i, \beta)) = f(B(v_j, \beta))$.
    Such nodes exist for every neighborhood of inputs by the pigeon hole principle as there are only $|\Gamma_\Pi|^{2r+1}$ possible output neighborhoods but there are $|\Gamma_\Pi|^{2r+1} + 1$ copies of each input neighborhood by construction. The intuition here is that $f(\iota)$ is the \emph{canonical} way to label neighborhood $\iota$, and we can extract such a canonical labeling by simulating~$\algoa$.
    
    Using the function $f$, we can now construct a locality-$O(\log* n)$ \LOCAL algorithm $\algoa'$ for solving $\Pi$ as follows:
    \begin{enumerate}
        \item Algorithm $\algoa'$ constructs a distance-$2\beta$ coloring of the path or cycle; this can be done with locality~$O(\log* n)$~\cite{cole86deterministic}.
        The algorithm then uses this coloring to construct a $(2\beta, 2\beta)$-ruling set $R'$ of the input graph by using the above-constructed coloring as a schedule for the simple greedy algorithm; this can be done with constant locality.
        The algorithm then prunes the nodes that are at distance less than $\beta$ from the endpoint of the path from set $R'$ to get a set $R$.
        Note that after the pruning, set $R$ is still a $(2\beta, 5\beta)$-ruling set of the input path or cycle.
        
        \item Algorithm $\algoa'$ uses $f$ to label the radius-$r$ neighborhood of each node $v \in R$ by $f(B(v, \beta))$; this can be done with locality $\beta+r$.
        Note that by the construction of $R$, each pair of nodes in $R$ is at least at a distance of $2\beta$ from each other. Thus, the radius-$\beta$ neighborhoods of the nodes are disjoint.
        
        \item The unlabeled nodes now reside in disjoint segments of length at most $5\beta$, sandwiched between radius-$r$ neighborhoods of nodes of $R$ (or the endpoints of the path).
        Algorithm $\algoa'$ can then use $\algoa$ to construct a locally valid output labeling for each disjoint segment in parallel as follows:
        
        Let $u_i$ and $u_j$ be nodes of $R$ marking the endpoints of the segment.
        Run $\algoa$ on the virtual path $P$ like in the construction of function $f$.
        Identify node $u_i$ (resp.\ $u_j$) with a node $v_i \in V(P)$ (resp.\ $v_j \in V(P)$) such that the radius-$\beta$ neighborhoods of inputs are the same.
        Connect the neighborhoods of $v_i$ and $v_j$ with new nodes on $P$ such that the path $u_i \leadsto u_j$ is isomorphic to the path $v_i \leadsto v_j$.
        Algorithm $\algoa'$ can now resume the execution of $\algoa$, and present the unlabeled nodes of path $v_i \leadsto v_j$ to $\algoa$ in some order to see how $\algoa$ would construct a valid labeling for this segment of the path.
        Since the labeling is locally valid on $v_i \leadsto v_j$ in path $P$, it must also be locally valid on $u_i \leadsto u_j$ in the original graph.
        
        We emphasize that the experiment of modifying $P$ and resuming the execution of $\algoa$ there is fully localized; all nodes along the path can do a similar thought experiment independently of each other.
    \end{enumerate}
    Composing all of the phases, we see that the locality of the \LOCAL algorithm $\algoa'$ is $O(\log* n)$.
\end{proof}

\section{LCL problems in rooted trees}
\label{sec:full-proof-rooted-trees}

In this section, we prove the following theorems:
\thmtreesuper*
\thmtreesub*

We restrict our attention to LCL problems in rooted trees without inputs.
We also restrict the family of graphs to be regular rooted trees. If the trees were not regular, the degrees of nodes could be used to encode input labels.
More formally, we use the following definition for LCL problems in rooted trees:
\begin{definition}[LCL problem in rooted trees \cite{balliu21rooted-trees}]
    \label{def:lcl-rooted-trees}
    An LCL problem $\Pi$ is a triple $(\delta, \Gamma, C)$ where:
    \begin{itemize}[noitemsep]
        \item $\delta$ is the number of allowed children,
        \item $\Gamma$ is a finite set of (output) labels,
        \item $C$ is a set of tuples of size $\delta + 1$ from $\Gamma^{\delta + 1}$ called \emph{allowed configurations}.
    \end{itemize}
    \noindent Each allowed configuration $(a \colon b_1 b_2 \ldots b_\delta) \in C$ states that a node with label $a$ is allowed to have children with labels $b_1, b_2, \ldots b_\delta$ in some order.
\end{definition}

In \cref{sec:rooted-tree-path-flexibility}, we define path-flexibility for labels of LCL problems that we need for proving $n^{\Omega(1)}$ lower bounds.
We then introduce the \twohalf-coloring problem in \cref{sec:rooted-tree-example} as an example of an LCL problem requiring locality~$\Omega(\sqrt{n})$ in the \LOCAL model, and show that it does so also in the \onlineLOCAL model.
In \cref{ssec:rooted-tree-equivalence-super}, we generalize this technique for all LCL problems requiring locality~$n^{\Omega(1)}$ in the \LOCAL model.
In \cref{ssec:rooted-tree-equivalence-sub}, we prove \cref{thm:rooted-tree-sublogarithmic} by showing how a locality-$o(\log n)$ \onlineLOCAL algorithm can be used to construct a locality-$O(\log* n)$ \LOCAL algorithm.

\subsection{LCL problems as automata: path-forms and path-flexibility}
\label{sec:rooted-tree-path-flexibility}

The concept of path-flexibility is based on previous work on LCL problems on paths and cycles~\cite{chang21automata-theoretic}.
Even though we are interested in rooted trees, we can still use results on directed paths by defining path-form for LCL problems.
The path-form of an LCL problem can be seen as a relaxation of the original problem:
Instead of requiring all children to have a compatible combination of labels, the path-form relaxes this condition by requiring that every child individually needs to have a valid label regardless of which labels its siblings have.
\begin{definition}[path-form \cite{balliu21rooted-trees}]
    Let $\Pi = (\delta, \Gamma, C)$ be an LCL problem in rooted trees.
    The path-form of $\Pi$ is an LCL problem $\pathform{\Pi} = (1, \Gamma, C')$ on directed paths.
    A configuration $(a : b)$ belongs to $C'$ if and only if there exists a configuration $(a : b_1, b_2, \ldots, b_\delta)$ in $C$ with $b = b_i$ for some $i$.
\end{definition}

Path-flexibility of label $\gamma$ in problem $\Pi$ is defined using an automaton $\machine(\Pi)$ associated with the path-form $\pathform{\Pi}$ of problem $\Pi$.
The intuition with the automaton is that a sequence of labels on a directed path is valid for a problem $\pathform{\Pi}$ exactly when there exists a sequence of transitions in $\machine(\Pi)$ visiting the states corresponding to the labels in the same order.
\begin{definition}[automaton associated with path-form \cite{chang21automata-theoretic}]
    Let $\Pi = (\delta, \Gamma, C)$ be an LCL problem, and let $\pathform{\Pi} = (1, \Gamma, C')$ be its path-form.
    The automaton $\machine(\Pi)$ associated with the path-form $\pathform{\Pi}$ is a non-deterministic unary semiautomaton.
    The set of states is $\Gamma$, and the automaton has a transition $a \rightarrow b$ if the configuration $(a : b)$ exists in $C'$.
\end{definition}

Flexibility is a property of a state of an automaton, and is directly related to path-flexibility of a label.
The intuition with path-flexibility is that, if a label is path-flexible, then it can be used to label far-away nodes on a path without knowing the exact distance between the nodes.
If, on the other hand, a state is path-inflexible, then the labels can exists only at certain distances from each other.
For example, consider the problem of 2-coloring a path with labels \texttt{1} and \texttt{2}.
Two nodes having label \texttt{1} can appear on a correctly colored path only if their distance is even, otherwise the path must be ill-colored.
Therefore, the label \texttt{1} is path-inflexible.
\begin{definition}[flexible state \cite{brandt17grid-lcl}]
    Let $\gamma$ be a state of the automaton $\machine(\Pi)$.
    The state $\gamma$ is flexible if there exists a constant $K$ such that for every $d \ge K$, there exists a walk $\gamma \leadsto \gamma$ in $\machine(\Pi)$ with length $d$.
\end{definition}
\begin{definition}[path-flexibility \cite{balliu21rooted-trees}]
    Let $\Pi$ be an LCL problem, and let $\gamma$ be one of its labels.
    The label $\gamma$ is called path-flexible if the state $\gamma$ is flexible in the automaton $\machine(\Pi)$.
    Otherwise, the label is called path-inflexible.
\end{definition}

Path-flexibility can also be generalized for multiple labels.
The intuition is the same as with individual path-flexible labels:
Nodes having labels from a path-inflexible pair can reside only at certain distances apart from each other.
\begin{definition}[path-flexible pair]
    Let $\Pi$ be an LCL problem, and let $(\gamma_1, \gamma_2)$ be a pair of labels of $\Pi$.
    The pair is path-flexible if there exists a constant $K$ such that for every $d \ge K$, there exists walks $\gamma_1 \leadsto \gamma_2$ and $\gamma_2 \leadsto \gamma_1$ of length $d$ in the automaton $\machine(\Pi)$.
\end{definition}

Path-flexibility of a pair is directly connected to path-flexibility of the individual labels that it consists of, as formalized in the following lemma:
\begin{lemma}
    Let $\Pi$ be an LCL problem.
    A label pair $(\gamma_1, \gamma_2)$ is path-flexible if and only if both $\gamma_1$ and $\gamma_2$ belong to the same strongly connected component of $\machine(\Pi)$ and are separately path-flexible.
\end{lemma}
\begin{proof}
    It is easy to see that if the labels belong to the same strongly connected component and both are path-flexible, then also the pair is path-flexible.
    The other direction is almost as easy:
    If the label pair is path-flexible, then there exists path-flexible walks from $\gamma_1$ to $\gamma_2$, and vice versa.
    Hence the labels must belong to the same strongly connected component. The labels must also be path-flexible as we can combine two walks in opposite directions to get a walk of an arbitrary length from both labels to themselves.
\end{proof}

If a pair of labels is path-inflexible, then it is easy to find a way to combine two fragments containing those labels such that no labeling exists for the resulting graph.
This happens when the distance between the labeled nodes is such that the corresponding walk does not exist in $\machine(\Pi)$.
This is formalized in the following lemma.
\begin{lemma}
    \label{lemma:rooted-tree-combinations}
    Let $(\gamma_1, \gamma_2)$ be a path-inflexible pair in $\machine(\Pi)$.
    Then, for every pair $p_1, p_2 \in \mathbb{N}$, not all of the following walks can exist in $\machine(\Pi)$:
    \begin{enumerate}[noitemsep, label=(\roman*)]
        \item\label{walk-p10} a walk $\gamma_1 \leadsto \gamma_2$ of length $p_1$,
        \item\label{walk-p11} a walk $\gamma_1 \leadsto \gamma_2$ of length $p_1 + 1$,
        \item\label{walk-p20} a walk $\gamma_2 \leadsto \gamma_1$ of length $p_2$, and
        \item\label{walk-p21} a walk $\gamma_2 \leadsto \gamma_1$ of length $p_2 + 1$.
    \end{enumerate}
\end{lemma}
\begin{proof}
    Assume for contradiction that all of the walks existed.
    Then the claim is that there exists some $K$ such that for every $d \ge K$ there exists walks $\gamma_1 \leadsto \gamma_2$ and $\gamma_2 \leadsto \gamma_1$ of length exactly $d$, and therefore $(\gamma_1, \gamma_2)$ is a path-flexible pair.

    Such walks can be constructed.
    Consider the following walks from $\gamma_1$ back to $\gamma_1$:
    \begin{enumerate}[noitemsep]
        \item[$W_1$] Walk along \ref{walk-p10} to go from $\gamma_1$ to $\gamma_2$, and then take walk \ref{walk-p20} to go back to $\gamma_1$.
        \item[$W_2$] Walk along \ref{walk-p10} to go from $\gamma_1$ to $\gamma_2$, and then take walk \ref{walk-p21} to go back to $\gamma_1$.
    \end{enumerate}
    The walk $W_1$ has a length of $p_1 + p_2$, and $W_2$ has a length of $p_1 + p_2 + 1$.
    Combining the walks $W_1$ and $W_2$ repeatedly makes it possible to construct a walk from $\gamma_1$ back to itself for any length of at least $K' = (p_1 + p_2)(p_1 + p_2 - 1)$ \cite{10.1007/978-3-540-85780-8_5}.

    A similar construction can be used to construct a self-walk from $\gamma_2$ back to itself for any length of at least $K'$.
    Combining these with walks \ref{walk-p10} and \ref{walk-p20}, it is possible to construct walks $\gamma_1 \leadsto \gamma_2$ and $\gamma_2 \leadsto \gamma_1$ for any length $d$ of at least $K = K' + \max(p_1, p_2)$.
    Hence the pair $(\gamma_1, \gamma_2)$ is path-flexible.
\end{proof}

\subsection{\texorpdfstring{\boldmath\twohalf-coloring takes $\Omega(\sqrt{n})$ in \onlineLOCAL}{2½-coloring takes Ω(√n) in \onlineLOCAL}}
\label{sec:rooted-tree-example}

In this section, we will present an LCL problem that requires locality~$\Omega(\sqrt{n})$ in the \onlineLOCAL model.
The problem is called \twohalf-coloring and was first introduced by \citet*{Chang2017ATH}.

Informally, the \twohalf-coloring problem requires 2-coloring the tree near the root with labels \texttt{A} and \texttt{B}, and near the leaves with \texttt{1} and \texttt{2}.
The different 2-colored parts can be combined using an intermediate label \texttt{X}.

Formally, the \twohalf-coloring problem $\probtwohalf$ is defined as follows:
\begin{definition}[\twohalf-coloring problem $\probtwohalf$]
    Problem $\probtwohalf$ has five labels: \texttt{A}, \texttt{B}, \texttt{X}, \texttt{1}, and \texttt{2}.
    Nodes with label \texttt{A} can have any combination of children labeled with \texttt{B} and \texttt{X}.
    Similarly nodes with label \texttt{B} can have any combination of children labeled with \texttt{A} and \texttt{X}.
    Nodes with label \texttt{X} need to have at least one child with label \texttt{1}, and the other child can have any label other than \texttt{X}.
    Nodes with label \texttt{1} need to have two children labeled with \texttt{2}, and nodes with label \texttt{2} need to have two children labeled with~\texttt{1}.
\end{definition}

The reason for naming this problem \twohalf-coloring is that any valid 2-coloring is also a valid \twohalf-coloring, and a valid \twohalf-coloring can be transformed into a valid 3-coloring by mapping $\texttt{A} \mapsto \texttt{1}$ and $\texttt{B} \mapsto \texttt{2}$.
Hence the \twohalf-coloring problem is at least as hard as the 3-coloring problem, but at most as hard as the 2-coloring problem.

\begin{figure}
    \centering
    \resizebox*{0.5\textwidth}{!}{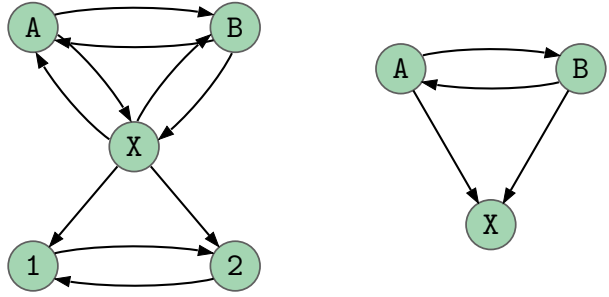}
    \caption[Automata associated with path forms of $\probtwohalf$ and $\probtwohalf'$.]{
        The automaton associated with the path-form of problem $\probtwohalf$ on the left, and the automaton associated with the path from of $\probtwohalf'$ on the right.
        On the left, the states \texttt{A}, \texttt{B} and \texttt{X} are path-flexible as there exists walks from them to back themselves for all lengths of at least 4.
        The states \texttt{1} and \texttt{2} are path-inflexible as only self walks with even length are possible.
        On the right, all states \texttt{A}, \texttt{B} and \texttt{X} are path-inflexible.
        States \texttt{A} and \texttt{B} have only even-length self walks.
        The state \texttt{X} has no self walks.
    }
    \Description{}
    \label{fig:twohalf-automata}
\end{figure}

By analyzing the problem description of $\probtwohalf$, and the automaton associated with its path-form shown in \cref{fig:twohalf-automata}, the following observation can be made:
\begin{observation}
    \label{obs:twohalf-inflexibility}
    The labels \textup{\texttt{1}} and \textup{\texttt{2}} are path-inflexible in problem $\probtwohalf$.
    This is because any path between two nodes labeled with labels \textup{\texttt{1}} and \textup{\texttt{2}} must form a valid 2-coloring.
    Denote the set of path-inflexible labels by
    $\Gamma_1 = \{ \textup{\texttt{1}},\ \textup{\texttt{2}} \}$.
\end{observation}
Problem $\probtwohalf$ can be restricted by removing these labels from the set of possible labels.
This gives the following restricted problem $\probtwohalf'$:
\begin{definition}[restricted \twohalf-coloring problem $\probtwohalf'$]
    Problem $\probtwohalf'$ has three labels: \texttt{A}, \texttt{B}, and \texttt{X}.
    Nodes with label \texttt{A} can have any combination of children labeled with \texttt{B} and \texttt{X}.
    Similarly, nodes with label \texttt{B} can have any combination of children labeled with \texttt{A} and \texttt{X}.
    Nodes with label \texttt{X} cannot have any children.
\end{definition}
By analyzing the problem description of $\probtwohalf'$, and the automaton associated with its path-form shown in \cref{fig:twohalf-automata}, we arrive at the following observation:
\begin{observation}
    \label{obs:restricted-twohalf-inflexibility}
    All labels \textup{\texttt{A}}, \textup{\texttt{B}} and \textup{\texttt{X}} are path-inflexible in problem $\probtwohalf'$.
    This is because the label \textup{\texttt{X}} cannot be used to label any internal nodes of the tree, and hence this is (almost) a regular 2-coloring problem with labels \textup{\texttt{A}} and \textup{\texttt{B}}.
    Denote the set of inflexible labels by
    $ \Gamma_2 = \{ \textup{\texttt{A}},\ \textup{\texttt{B}},\ \textup{\texttt{X}} \} $.
\end{observation}

The lower-bound construction for \twohalf-coloring relies on the fact that an algorithm solving $\probtwohalf$ cannot use labels \texttt{1} and \texttt{2} to label two nodes without seeing how the nodes are connected to each other.
Otherwise the adversary could alter the length of the path connecting the nodes such that the labeling cannot be completed.
This is because labels \texttt{1} and \texttt{2} are path-inflexible, as noted in \cref{obs:twohalf-inflexibility}.

Nevertheless any sufficiently local algorithm can be forced to use labels \texttt{1} and \texttt{2} in far-away parts of the graph using the following trick:
Show the algorithm two fragments of the graph such that the algorithm does not know how they are connected.
Based on the previous observation, the algorithm must use labels \texttt{A}, \texttt{B} and \texttt{X} to label the shown nodes.

By \cref{obs:restricted-twohalf-inflexibility}, labels \texttt{A}, \texttt{B} and \texttt{X} are path-inflexible in the restricted \twohalf-coloring problem $\probtwohalf'$.
Hence the adversary can connect the fragments in such a way that there exists no valid labeling for the path connecting the two fragments that satisfies the restricted problem~$\probtwohalf'$. Therefore, the algorithm must use labels \texttt{1} and \texttt{2} somewhere along the connecting path.
The construction can be repeated to create another node with label \texttt{1} or \texttt{2}.
By being careful with the construction, it can be ensured that the algorithm has not seen how the two nodes having labels \texttt{1} and \texttt{2} are connected to each other. Hence the adversary can force the distance between them to be such that there exists no valid labeling for $\probtwohalf$ for the resulting tree.

Informally, the components that are shown to the algorithm are formed in the following way:
Let $P$ be a directed path with $x$ nodes from $t$ to $s$, and let $c$ be the middlemost node on the path.
Collectively call the nodes on path $P$ \emph{layer-2 nodes}, and call node $t$ the \emph{connector node}.
Identify each node along path $P$ as the root of another path of $x+1$ nodes, and call the newly-created nodes \emph{layer-1 nodes}.
As the last step, add one child node to each layer-1 node. Call these added nodes \emph{layer-0 nodes}.
We call the tree formed in this way $T^x_2$.
\cref{fig:2-layer-tree} shows an example of tree $T^5_2$.

\begin{figure}
    \centering
    \resizebox*{0.8\textwidth}{!}{\fontsize{11}{11}\selectfont
\begingroup%
  \makeatletter%
  \providecommand\color[2][]{%
    \errmessage{(Inkscape) Color is used for the text in Inkscape, but the package 'color.sty' is not loaded}%
    \renewcommand\color[2][]{}%
  }%
  \providecommand\transparent[1]{%
    \errmessage{(Inkscape) Transparency is used (non-zero) for the text in Inkscape, but the package 'transparent.sty' is not loaded}%
    \renewcommand\transparent[1]{}%
  }%
  \providecommand\rotatebox[2]{#2}%
  \newcommand*\fsize{\dimexpr\f@size pt\relax}%
  \newcommand*\lineheight[1]{\fontsize{\fsize}{#1\fsize}\selectfont}%
  \ifx\svgwidth\undefined%
    \setlength{\unitlength}{388bp}%
    \ifx\svgscale\undefined%
      \relax%
    \else%
      \setlength{\unitlength}{\unitlength * \real{\svgscale}}%
    \fi%
  \else%
    \setlength{\unitlength}{\svgwidth}%
  \fi%
  \global\let\svgwidth\undefined%
  \global\let\svgscale\undefined%
  \makeatother%
  \begin{picture}(1,0.79381443)%
    \lineheight{1}%
    \setlength\tabcolsep{0pt}%
    \put(0,0){\includegraphics[width=\unitlength,page=1]{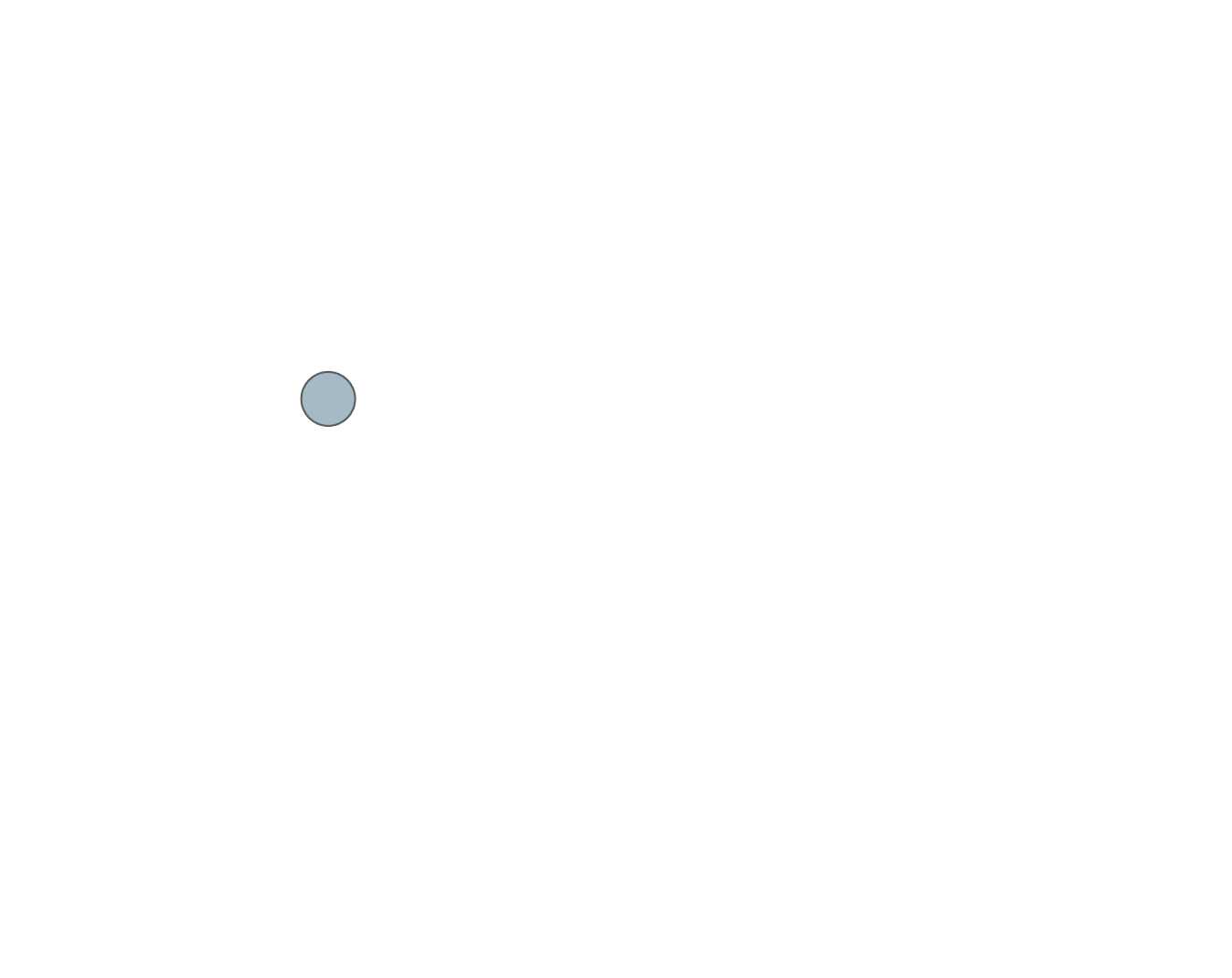}}%
    \put(0.26657284,0.46188748){\color[rgb]{0,0,0}\makebox(0,0)[t]{\lineheight{1.25}\smash{\begin{tabular}[t]{c}$t$\end{tabular}}}}%
    \put(0,0){\includegraphics[width=\unitlength,page=2]{example-layered-tree.pdf}}%
    \put(0.61672,0.60777785){\color[rgb]{0,0,0}\makebox(0,0)[t]{\lineheight{1.25}\smash{\begin{tabular}[t]{c}$c$\end{tabular}}}}%
    \put(0,0){\includegraphics[width=\unitlength,page=3]{example-layered-tree.pdf}}%
    \put(0.9668671,0.75366817){\color[rgb]{0,0,0}\makebox(0,0)[t]{\lineheight{1.25}\smash{\begin{tabular}[t]{c}$s$\end{tabular}}}}%
    \put(0,0){\includegraphics[width=\unitlength,page=4]{example-layered-tree.pdf}}%
  \end{picture}%
\endgroup%
}
    \caption{
        An example of a layered tree $T^5_2$.
        The darkest nodes belong to layer 2, the lighter nodes belong to layer 1, and the white leaf nodes belong to layer 0.
        Node $s$ is the root, node $t$ is the connector, and node $c$ is the middlemost node on the core path.
    }
    \Description{}
    \label{fig:2-layer-tree}
\end{figure}

We can now prove the lower bound for the \twohalf-coloring problem.
The proof is based on a simulation of an \onlineLOCAL algorithm $\algoa$ which supposedly solves the \twohalf-coloring problem with locality $o(\sqrt{n})$.
We simulate algorithm $\algoa$ on disjoint graph fragments, and only after the algorithm has committed the labels for some nodes, we decide how the fragments are connected to each other.
In the following, we will show how we can use the commitments of the algorithm to construct an input instance for which the algorithm must fail.
We will also illustrate a concrete example of the construction for an algorithm with locality $T(n) = 1$ with parameter $x = 5$.

Assume for contradiction that there existed an \onlineLOCAL algorithm $\algoa$ solving problem $\probtwohalf$ with locality $T(n) = o(\sqrt{n})$.
We can now construct a failing instance $G_\algoa$ as follows:
\begin{enumerate}
    \item Construct $4$ copies of tree $T^x_2$, namely $G_1, G_2, G_3, G_4$, with sufficiently large constant $x \gg 2T(n)$, where $n$ is the size of the graph.
    This is possible because each copy of tree $T^x_2$ has $x + 2x^2$ nodes, and hence the whole graph has $4x + 8x^2 < 16x^2$ nodes.
    By assumption, the locality is $T(n) = o(\sqrt{n})$, and hence there exists some $x$ such that
    $ x \gg 2T(16x^2) $
    holds.

    \begin{figure}
        \centering
        \resizebox*{0.8\textwidth}{!}{\fontsize{9}{9}\selectfont
\begingroup%
  \makeatletter%
  \providecommand\color[2][]{%
    \errmessage{(Inkscape) Color is used for the text in Inkscape, but the package 'color.sty' is not loaded}%
    \renewcommand\color[2][]{}%
  }%
  \providecommand\transparent[1]{%
    \errmessage{(Inkscape) Transparency is used (non-zero) for the text in Inkscape, but the package 'transparent.sty' is not loaded}%
    \renewcommand\transparent[1]{}%
  }%
  \providecommand\rotatebox[2]{#2}%
  \newcommand*\fsize{\dimexpr\f@size pt\relax}%
  \newcommand*\lineheight[1]{\fontsize{\fsize}{#1\fsize}\selectfont}%
  \ifx\svgwidth\undefined%
    \setlength{\unitlength}{410bp}%
    \ifx\svgscale\undefined%
      \relax%
    \else%
      \setlength{\unitlength}{\unitlength * \real{\svgscale}}%
    \fi%
  \else%
    \setlength{\unitlength}{\svgwidth}%
  \fi%
  \global\let\svgwidth\undefined%
  \global\let\svgscale\undefined%
  \makeatother%
  \begin{picture}(1,1.40243902)%
    \lineheight{1}%
    \setlength\tabcolsep{0pt}%
    \put(0,0){\includegraphics[width=\unitlength,page=1]{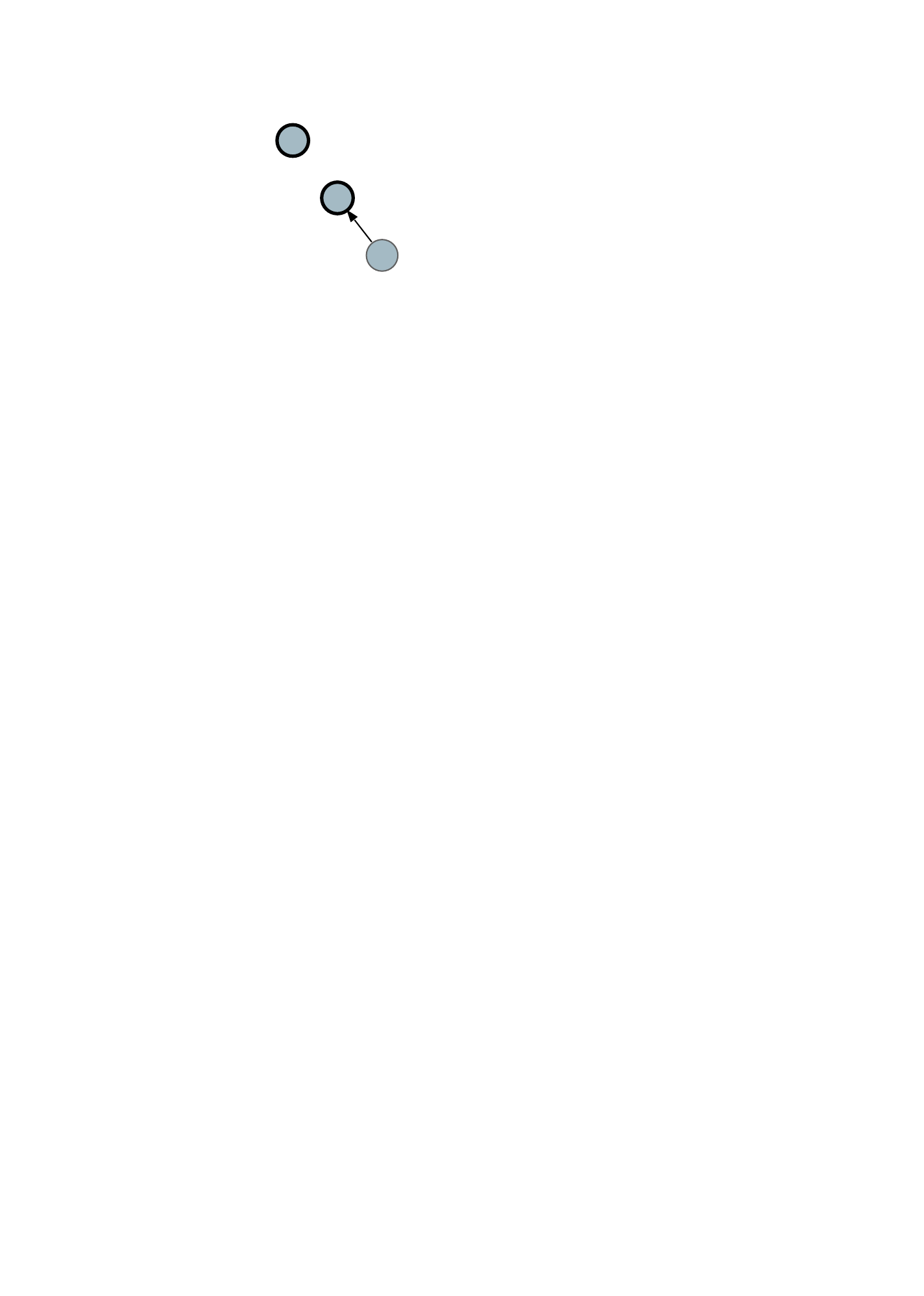}}%
    \put(0.31694895,1.24283035){\color[rgb]{0,0,0}\makebox(0,0)[t]{\lineheight{1.25}\smash{\begin{tabular}[t]{c}\texttt{1}\end{tabular}}}}%
    \put(0,0){\includegraphics[width=\unitlength,page=2]{superlogarithmic-example-1.pdf}}%
    \put(0.87647135,1.24283035){\color[rgb]{0,0,0}\makebox(0,0)[t]{\lineheight{1.25}\smash{\begin{tabular}[t]{c}\texttt{X}\end{tabular}}}}%
    \put(0,0){\includegraphics[width=\unitlength,page=3]{superlogarithmic-example-1.pdf}}%
    \put(0.31694895,0.51749829){\color[rgb]{0,0,0}\makebox(0,0)[t]{\lineheight{1.25}\smash{\begin{tabular}[t]{c}\texttt{A}\end{tabular}}}}%
    \put(0,0){\includegraphics[width=\unitlength,page=4]{superlogarithmic-example-1.pdf}}%
    \put(0.87647135,0.51749829){\color[rgb]{0,0,0}\makebox(0,0)[t]{\lineheight{1.25}\smash{\begin{tabular}[t]{c}\texttt{B}\end{tabular}}}}%
    \put(0,0){\includegraphics[width=\unitlength,page=5]{superlogarithmic-example-1.pdf}}%
  \end{picture}%
\endgroup%
}
        \caption{
            Trees $G_1$, $G_2$, $G_3$ and $G_4$.
            The algorithm has labeled the center nodes on the core paths with \texttt{1}, \texttt{X}, \texttt{A} and \texttt{B}.
            The neighborhoods the algorithm has seen are visualized by thicker lines around nodes.
        }
        \Description{}
        \label{fig:superlogarithmic-example-1}
    \end{figure}

    \item Reveal the center nodes on the core paths of the trees $G_i$ to algorithm $\algoa$.
    The algorithm must commit to some labels for these nodes without seeing the roots or the connector nodes of the trees.
    This means that the algorithm does not know how the trees $G_i$ are connected to each other.
    This is because $x$ is larger than the diameter of the view that is revealed to the algorithm.
    \cref{fig:superlogarithmic-example-1} visualizes an example of this situation.

    \item Consider now the following cases.
    In each one of them, we can force algorithm~$\algoa$ to fail to produce a valid labeling:
    \begin{enumerate}[label=(\roman*)]
        \item \label{itm:two-from-small-set} \emph{Algorithm $\algoa$ labels at least two of the center nodes of the trees with labels from set $\Gamma_1$.}

        Let $v_1$ and $v_2$ be two distinct center nodes labeled with labels $\gamma_1, \gamma_2 \in \Gamma_1$, respectively.
        Without loss of generality, we may assume that nodes $v_1$ and $v_2$ are the center nodes of trees $G_1$ and $G_2$.
        By \cref{obs:twohalf-inflexibility}, the labels $\gamma_1$ and $\gamma_2$ are path-inflexible in $\probtwohalf$, and hence the pair $(\gamma_1, \gamma_2)$ is a path-inflexible pair by \cref{lemma:rooted-tree-combinations}.

        Consider two different ways of combining the trees $G_1$ and $G_2$:
        The root of tree $G_2$ can be identified with the connector node of $G_1$, or, alternatively, the root of tree $G_2$ can be made a child of the connector node.
        \cref{fig:superlogarithmic-example-2} visualizes both of these cases.
        In the former case the length of the path between $v_1$ and $v_2$ is $p$ and in the latter it is $p+1$.
        As $p$ and $p+1$ have different parity, only one of the walks $\gamma_1 \leadsto \gamma_2$ of length $p$ and $p+1$ can exist in $\machine(\probtwohalf)$.
        We can choose the option for which the walk of that length does not exist.
        This implies that there is no way to label the path between $v_1$ and $v_2$ such that the labeling would be valid according to the \twohalf-coloring problem $\probtwohalf$.
        Hence the algorithm must fail.

        \begin{figure}[p]
            \centering
            \resizebox*{0.8\textwidth}{!}{\fontsize{9}{9}\selectfont
\begingroup%
  \makeatletter%
  \providecommand\color[2][]{%
    \errmessage{(Inkscape) Color is used for the text in Inkscape, but the package 'color.sty' is not loaded}%
    \renewcommand\color[2][]{}%
  }%
  \providecommand\transparent[1]{%
    \errmessage{(Inkscape) Transparency is used (non-zero) for the text in Inkscape, but the package 'transparent.sty' is not loaded}%
    \renewcommand\transparent[1]{}%
  }%
  \providecommand\rotatebox[2]{#2}%
  \newcommand*\fsize{\dimexpr\f@size pt\relax}%
  \newcommand*\lineheight[1]{\fontsize{\fsize}{#1\fsize}\selectfont}%
  \ifx\svgwidth\undefined%
    \setlength{\unitlength}{385bp}%
    \ifx\svgscale\undefined%
      \relax%
    \else%
      \setlength{\unitlength}{\unitlength * \real{\svgscale}}%
    \fi%
  \else%
    \setlength{\unitlength}{\svgwidth}%
  \fi%
  \global\let\svgwidth\undefined%
  \global\let\svgscale\undefined%
  \makeatother%
  \begin{picture}(1,0.85194805)%
    \lineheight{1}%
    \setlength\tabcolsep{0pt}%
    \put(0,0){\includegraphics[width=\unitlength,page=1]{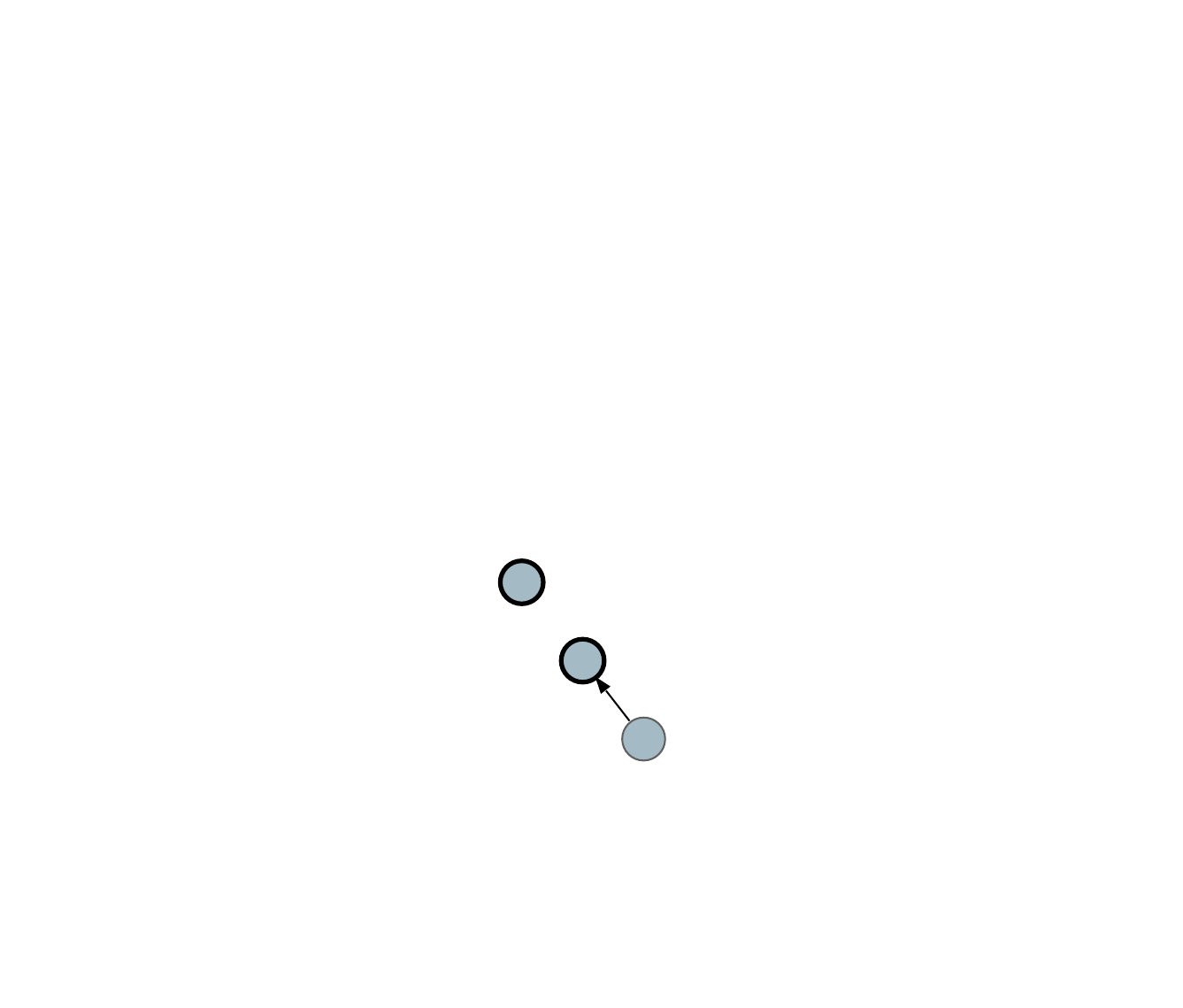}}%
    \put(0.44112902,0.35169405){\color[rgb]{0,0,0}\makebox(0,0)[t]{\lineheight{1.25}\smash{\begin{tabular}[t]{c}\texttt{1}\end{tabular}}}}%
    \put(0,0){\includegraphics[width=\unitlength,page=2]{superlogarithmic-example-2.pdf}}%
    \put(0.18355062,0.68287733){\color[rgb]{0,0,0}\makebox(0,0)[t]{\lineheight{1.25}\smash{\begin{tabular}[t]{c}\texttt{1}\end{tabular}}}}%
    \put(0,0){\includegraphics[width=\unitlength,page=3]{superlogarithmic-example-2.pdf}}%
    \put(0.86798576,0.3847968){\color[rgb]{0,0,0}\makebox(0,0)[t]{\lineheight{1.25}\smash{\begin{tabular}[t]{c}\texttt{1}\end{tabular}}}}%
    \put(0,0){\includegraphics[width=\unitlength,page=4]{superlogarithmic-example-2.pdf}}%
    \put(0.66191788,0.64974862){\color[rgb]{0,0,0}\makebox(0,0)[t]{\lineheight{1.25}\smash{\begin{tabular}[t]{c}\texttt{1}\end{tabular}}}}%
    \put(0,0){\includegraphics[width=\unitlength,page=5]{superlogarithmic-example-2.pdf}}%
  \end{picture}%
\endgroup%
}
            \caption[Two ways to connect trees with labels from set $\Gamma_1$.]{
                Example for case~\ref{itm:two-from-small-set}.
                Two ways to connect trees with labels from set $\Gamma_1$; in this example, the algorithm has decided to label both center nodes with label \texttt{1}.
                Note that in this visualization the chain of layer-1 nodes has been made shorter to draw attention to layer 2; in reality there are more layer-1 and layer-0 nodes.
                On the left, the root of the lower tree has been made a child of the connector node of the upper tree, while on the right they have been identified as one node.
                The distance between labeled nodes is 5 in the left tree and 4 in the right tree.
                Because the distance between nodes labeled with \texttt{1} must be even, the labeling on the left tree cannot be completed.
            }
    \Description{}
    \label{fig:superlogarithmic-example-2}
        \end{figure}

        \item \label{itm:all-from-large-set} \emph{Algorithm $\algoa$ labels all center nodes of the trees with labels from set $\Gamma_2$.}

        In this case, two trees with labels from set $\Gamma_2$ can be combined into one tree with a label from set $\Gamma_1$.
        By repeating this for both pairs of trees, we can construct two trees with labels from set $\Gamma_1$ such that the algorithm has not seen how the trees are connected.
        A contradiction can then be derived using the construction from case~\ref{itm:two-from-small-set} on these trees.

        To see how to construct a tree with a label from set $\Gamma_1$, consider two of the trees $G_i$, say $G_1$ and $G_2$. The corresponding center nodes are $v_1$ and $v_2$, and their labels are $\gamma_1$ and $\gamma_2$, respectively.
        By \cref{obs:restricted-twohalf-inflexibility}, labels $\gamma_1$ and $\gamma_2$ are path-inflexible in $\probtwohalf'$.
        Using a similar construction as in case~\ref{itm:two-from-small-set}, we can construct a tree in which algorithm $\algoa$ cannot solve the restricted problem $\probtwohalf'$.
        In particular, algorithm $\algoa$ must fail to solve $\probtwohalf'$ on the path connecting $v_1$ and $v_2$.

        However, algorithm $\algoa$ is not actually trying to solve $\probtwohalf'$ but $\probtwohalf$, and hence it can also use labels from set $\Gamma_1$ to label some of the revealed nodes.
        By revealing the nodes between $v_1$ and $v_2$, as well as their children, to algorithm $\algoa$, we can force the algorithm to use a label from set $\Gamma_1$ for some revealed node $v$ without showing the root of the tree to the algorithm.
        Moreover, revealing node $v$ does not let the algorithm see the whole subtree rooted at $v$.
        In particular, the node at the end of a path consisting of only layer-1 nodes is not shown to the algorithm. Therefore, we can use that node to connect the tree to another tree.

        Note that this is the case where the structure of the trees $G_i$ comes into play.
        The layered tree structure ensures that when node $v$ is revealed to the algorithm, the algorithm has not seen all layer-1 nodes in the subtree rooted at $v$.
        This is because the algorithm has locality $T(n)$, but the length of the path consisting of layer-1 nodes has length $x \gg 2T(n)$.
        The layered tree structure is also needed for proving the general result in \cref{ssec:rooted-tree-equivalence-super}.

        \Cref{fig:superlogarithmic-example-3} shows an example of how two trees can be connected to force the algorithm to use a label from set $\Gamma_1$.

        \begin{figure}[p]
            \centering
            \resizebox*{0.5\textwidth}{!}{\fontsize{9}{9}\selectfont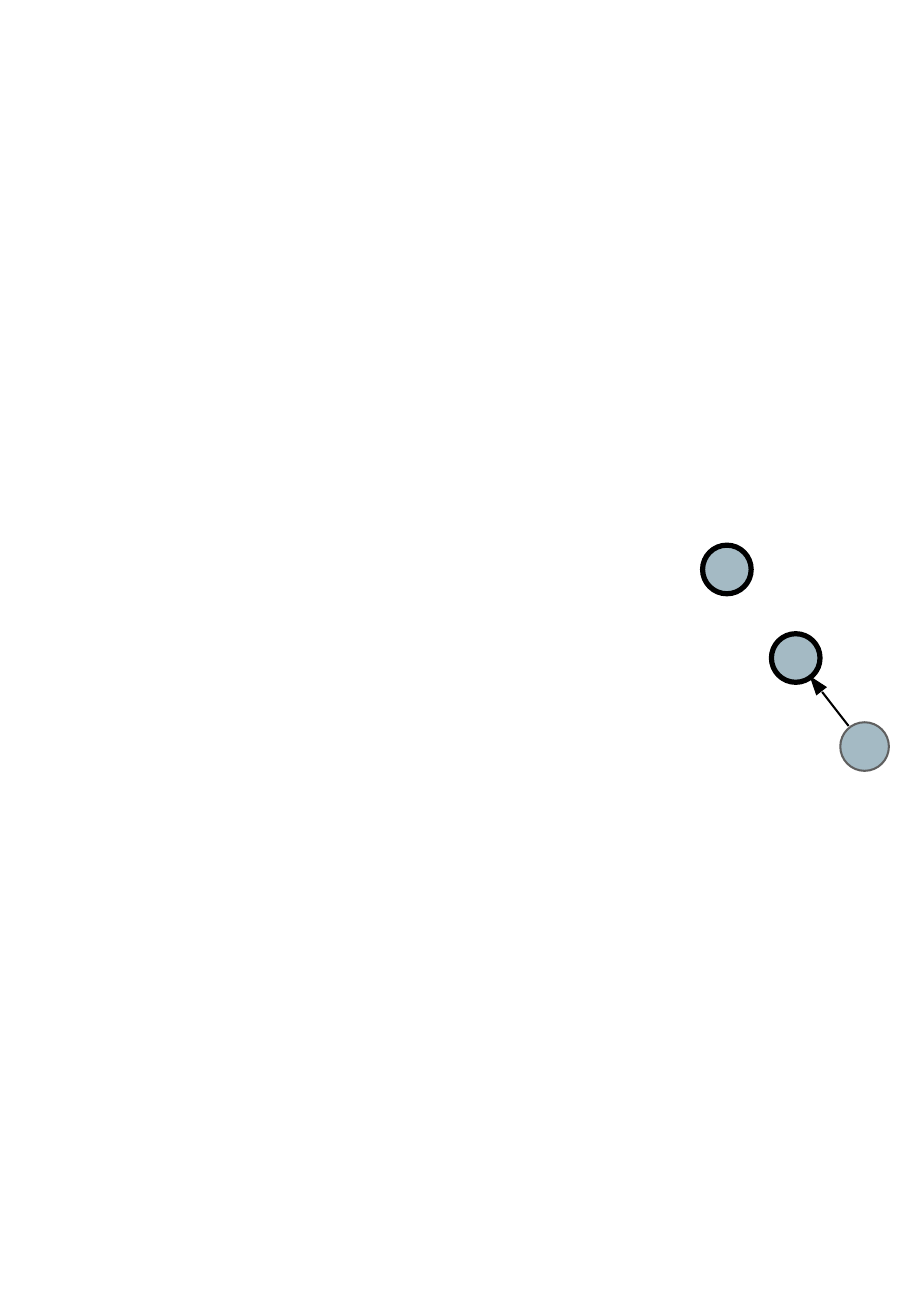}
            \caption{
                Example for case~\ref{itm:all-from-large-set}.
                We have connected trees $G_3$ and $G_4$, having labels \texttt{A} and \texttt{B}, in such a way that the connecting path cannot be labeled according to the restricted problem $\probtwohalf'$.
                The algorithm has labeled the nodes on the connecting path, and the children of those nodes.
                Because there is no valid labeling according to the restricted problem $\probtwohalf'$, the algorithm had to use a label from set $\Gamma_1$ to label at least one node.
                In this case, the algorithm has decided to use label \texttt{1}.
                The node labeled with \texttt{1} is node $v$ from the text, and the node marked with $t$ is the connector node in the subtree rooted at $v$.
            }
            \Description{}
            \label{fig:superlogarithmic-example-3}
        \end{figure}

        \item \label{itm:both-sets-present} \emph{Algorithm $\algoa$ labels exactly one of the center nodes of the trees with label from set $\Gamma_1$, and the rest of the center nodes from set $\Gamma_2$.}
        
        In this case, two of the trees whose center nodes have been labeled with labels from set $\Gamma_2$ can be combined like in case~\ref{itm:all-from-large-set} to produce a tree with a label from set $\Gamma_1$.
        To force the algorithm to fail, this resulted tree can then be combined with the original tree whose center node is from set $\Gamma_1$, just like in case~\ref{itm:two-from-small-set}.
        See Figure~\ref{fig:superlogarithmic-example-4} for an example.

        \begin{figure}
            \centering
            \resizebox*{0.5\textwidth}{!}{\fontsize{9}{9}\selectfont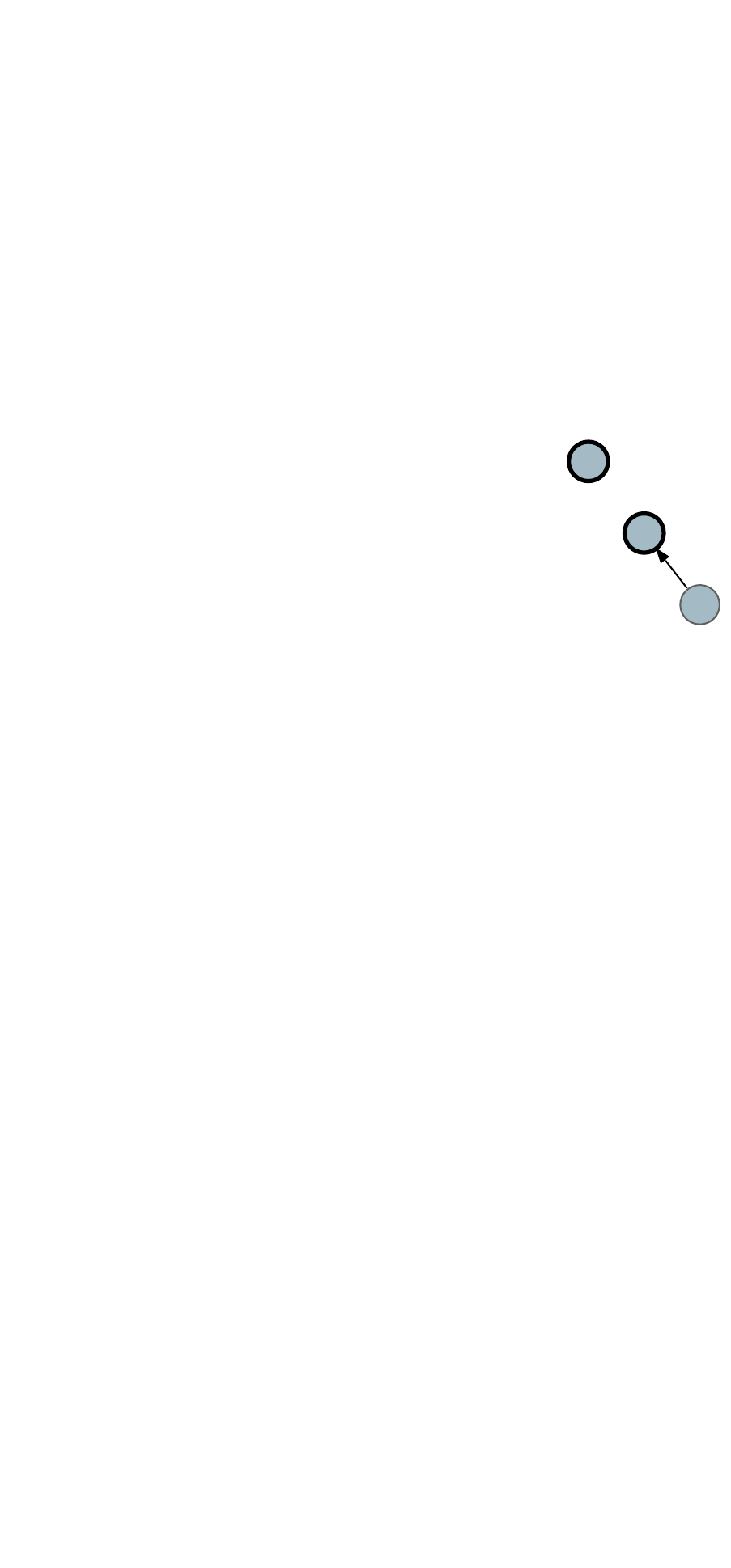}
            \caption{
                Example for case~\ref{itm:both-sets-present}.
                In this case, we first combined trees $G_3$ and $G_4$ to get a node with label \texttt{1}.
                We then made tree $G_1$ a child of the connector node of the newly-formed tree.
                The algorithm has tried to label all nodes on the path connecting the two labeled parts, but as there exists no valid labeling for the path, the algorithm has failed to label the node marked with \texttt{!}.
            }
            \Description{}
            \label{fig:superlogarithmic-example-4}
        \end{figure}
    \end{enumerate}
    
    This is an exhaustive list of all cases.
    As we can force algorithm~$\algoa$ to fail to produce a valid labeling in all cases, the assumption that such an algorithm can exist must be false.
    The only assumption made about algorithm~$\algoa$ is that it has locality $T(n) = o(\sqrt{n})$, and hence the \twohalf-coloring problem $\probtwohalf$ must require locality $\Omega(\sqrt{n})$.
\end{enumerate}

In the next section, we will generalize the techniques presented in this example to all LCL problems.

\flushfigures

\subsection{Equivalence in the super-logarithmic region}
\label{ssec:rooted-tree-equivalence-super}
It is known that problems requiring locality~$n^{\Omega(1)}$ in the \LOCAL model have a very specific structure~\cite{balliu21rooted-trees}.
In particular, such problems can be decomposed into hierarchical sequences of restricted problems by repeatedly removing path-inflexible labels from the problem.
More formally:
\begin{definition}[path-inflexible decomposition]
    Let $\Pi$ be an LCL problem.
    The \emph{path-inflexible decomposition} consists of a sequence of problems $(\Pi_0, \Pi_1, \ldots, \Pi_k)$ and a sequence of labels $(\Gamma_1, \Gamma_2, \ldots, \Gamma_k)$.
    The sequences are defined as follows:
    \begin{enumerate}[noitemsep]
        \item $\Pi_0 = \Pi$.
        \item For every $i \in [1, k]$, let $\Gamma_i$ be the set of path-inflexible labels in $\Pi_{i-1}$.
        \item For every $i \in [1, k]$, let $\Pi_i$ be the restriction of $\Pi$ to label set $\Gamma \setminus (\Gamma_1 \cup \Gamma_2 \cup \cdots \cup \Gamma_i)$.
        \item Problem $\Pi_k$ is either an empty problem, or all of its labels are path-flexible.
    \end{enumerate}        
\end{definition}

Note that every problem in the sequence $(\Pi_0, \Pi_1, \ldots, \Pi_k)$ is a restriction of the previous one.
In particular, a solution to problem $\Pi_i$ is a valid solution also to any problem $\Pi_j$ with $j \le i$.

It is known that if all labels of an LCL problem are path-flexible, then the LCL problem can be solved with locality $O(\log n)$ in the \LOCAL model \cite{balliu21rooted-trees}.
This holds also for the case where $\Pi$ itself has path-inflexible labels, but it has a restriction with only path-flexible labels.
In particular, if such restriction exists, then the last problem in its path-inflexible decomposition is such a restriction.

From now onwards, we are only interested in the case where the last problem $\Pi_k$ in the path-inflexible decomposition $(\Pi_0, \ldots, \Pi_k)$ is an empty problem.
We will show that in that case, the original problem requires locality $\Omega(n^{1/k})$ to be solved in the \onlineLOCAL model.
This lower bound also applies for the \LOCAL model.

The hierarchical structure of a problem is really useful when proving the $\Omega(n^{1/k})$ lower bound.
In fact, we use matching $k$-layer trees to construct a counter example for any locality-$o(n^{1/k})$ \onlineLOCAL algorithm.
A layered tree $T_k^x$ with $k$ layers and path length $x$ can be constructed as follows:
\begin{enumerate}
    \item If $k = 0$, then $T_0^x$ is a single node.
    The node belongs to layer $0$.
    \item Otherwise construct a directed path $(v_1, v_2, \ldots, v_x)$ with $x$ nodes.
    This path is called the \emph{core path} of $T_k^x$.
    For each node $v_i$, identify its children with the roots of $\delta-1$ copies of $T_{k-1}^x$.
    The nodes $v_i$ belong to layer $k$.
    Node $v_x$ is the root of $T_k^x$, and node $v_1$ is its \emph{connector node}.
\end{enumerate}
This layered structure is very useful.
In particular, for almost every node $v$, there exists a node belonging to a lower layer in the subtree rooted at $v$ that is far away from $v$.
This happens even though the size of the tree is only polynomial.
We formalize these observations in the following:
\begin{observation}
    \label{obs:rooted-tree-layers-continue}
    For every node $v$ in $T_k^x$ belonging to layer $i \ge 2$, there exists a layer-$(i-1)$ node $u$ in the subtree rooted at $v$ such that the distance between $v$ and $u$ is at least $x$.
\end{observation}
\begin{observation}
    Let $T_k^x$ be a layered tree.
    The tree consists of $O(x^k)$ nodes.
\end{observation}

We can now use the path-inflexible decomposition and layered trees to prove \cref{thm:rooted-tree-superlogarithmic}.
\begin{proof}[Proof of \cref{thm:rooted-tree-superlogarithmic}]
    Let $\Pi = (\delta, \Gamma, C)$ be an LCL problem, and let $(\Pi_0, \Pi_1, \ldots, \Pi_k)$ and $(\Gamma_1, \ldots, \Gamma_k)$ be its path-inflexible decomposition.
    If problem $\Pi_k$ is non-empty, then problem $\Pi$ can be solved with locality $O(\log n)$ in the \LOCAL model~\cite{balliu21rooted-trees}, and therefore also in the \onlineLOCAL model.
    Hence we will assume that $\Pi_k$ is an empty problem.
    Assume also that there existed an \onlineLOCAL algorithm $\algoa$ solving $\Pi$ with locality $T(n) = o(n^{1/k})$, where $k$ is the number of non-empty problems in the path-inflexible decomposition of $\Pi$.

    We will construct a rooted tree $G_\algoa$ on which algorithm~$\algoa$ fails.
    We start by constructing many copies of $T_k^x$ for sufficiently large $x$.
    Then we use algorithm~$\algoa$ to label one layer-$k$ node on each of the trees.
    Finally we combine the trees together in a way which ensures that no valid labeling exists for whole tree, given the existing labels.
    This proves that no such algorithm~$\algoa$ can exist.

    More formally, the construction works as follows:
    \begin{enumerate}
        \item\label{step:rooted-tree-superlogarithmic-tree-construction} Construct $2^{k+1}$ copies of layered tree $T_k^x$, with a sufficiently large parameter $x \gg 2T(n)$.
        Such large parameter $x$ exists, because each tree $T_k^x$ contains only $O(x^k)$ nodes, and by assumption $T(n) = o(n^{1/k})$.

        \item Use algorithm~$\algoa$ to label the middlemost nodes on the core paths of every tree.
        Because the parameter $x$ has been chosen to be large enough, algorithm~$\algoa$ does not see the ends of the core path.

        \item\label{step:rooted-tree-properties-of-C} Divide the trees into collections $C_1, \ldots, C_k$ based on the labels the algorithm has produced.
        If the algorithm used label $\gamma$ from set $\Gamma_i$ to label the node in tree $G$, put tree $G$ into collection~$C_i$.

        There are four properties which hold for every tree~$G$ in collection~$C_i$:
        \begin{enumerate}[noitemsep]
            \item There exists a node~$v$ in the tree~$G$ such that the label of $v$ belongs to set~$\Gamma_i$,
            \item the algorithm has not committed a label for any node whose layer is less than $i$,
            \item the algorithm has not seen the root nor the connector node of $G$, and
            \item the layer of the connector node is at least $i$.
        \end{enumerate}
        All these properties hold trivially for the initial trees.

        \item Iteratively combine two trees from collection~$C_i$ into one tree that can be added to collection~$C_j$ for some $j < i$.

        The combination of trees proceeds as follows:
        Consider two trees $A$ and $B$ from collection $C_i$.
        By the construction of $C_i$, there exist nodes $v_A$ and $v_B$ in $A$ and $B$, respectively, such that the labels of nodes $v_A$ and $v_B$ belong to set $\Gamma_i$.
        Moreover, the layers of $v_A$ and $v_B$ are at least $i$.

        There are four different ways in which trees $A$ and $B$ can be combined:
        \begin{enumerate}[noitemsep, label=(\alph*)]
            \item\label{itm:A-root-1} identify the root of $B$ with the connector node of $A$,
            \item\label{itm:B-root-1} identify the root of $A$ with the connector node of $B$,
            \item\label{itm:A-root-2} make the root of $B$ a child of the connector node of $A$, or
            \item\label{itm:B-root-2} make the root of $A$ a child of the connector node of $B$.
        \end{enumerate}
        Note that we could choose any one of these combinations without the algorithm knowing which one we chose.
        This is because the algorithm has not seen the roots or the connector nodes of either of the trees $A$ and $B$, and therefore we can change the structure of the trees in those neighborhoods freely.

        Among these four combinations, choose one for which there exists no valid labeling for problem $\Pi_{i-1}$.
        Such tree must exist by \cref{lemma:rooted-tree-combinations}.
        In particular, no valid labeling exists for the nodes and their children on the path connecting $v_A$ and $v_B$.
        Call the resulting tree $R$.
        Any attempt to label the nodes and their children on the path between $v_A$ and $v_B$ in $R$ according to problem $\Pi_{i-1}$ must fail.

        We now use algorithm~$\algoa$ to label all the nodes, and their children, on the path between $v_A$ and $v_B$ in $R$.
        Note that each of these nodes belongs to a layer whose index is at least $i-1$.
        By the choice of $R$, the algorithm cannot label all nodes using only labels present in problem~$\Pi_{i-1}$.
        Hence there must exist a node $v$ on the path between $v_A$ and $v_B$ such that it and its children do not form a valid configuration in problem~$\Pi_{i-1}$.
        This is possible only if either $v$ or at least one of its children has label $\gamma \in \Gamma_j$ for some $j < i$.
        Denote that node by $v'$.

        By \cref{obs:rooted-tree-layers-continue}, there exists a layer-$j$ node $u$ in the subtree rooted at $v$ such that the distance between $v$ and $u$ is at least $x$.
        Hence the distance between $v'$ and $u$ is at least $x-1 > T(n)$. This means that the algorithm has not seen node $u$ yet.
        Make node $u$ the connector node of tree $R$.

        It is easy to check that tree $R$ constructed in this way fulfills all properties of a tree of collection $C_j$, as described in step~\ref{step:rooted-tree-properties-of-C}.
        We can therefore remove trees $A$ and $B$ from collection $C_i$ and add tree $R$ to collection $C_j$.
        We repeat this step until collection $C_1$ contains at least two trees, at which point we move to the next step and force the algorithm to fail.

        \item In this last step, we force algorithm~$\algoa$ to fail.
        We do this by taking two trees $A$ and $B$ from collection $C_1$.
        By construction of collection $C_1$, there exist nodes $v_A$ and $v_B$ on trees $A$ and $B$ with labels $\gamma_A$ and $\gamma_B$ from set $\Gamma_1$, respectively.
        Recall that set $\Gamma_1$ is the set of path-inflexible labels in problem $\Pi_0 = \Pi$, and hence the label pair $(\gamma_A, \gamma_B)$ is a path-inflexible pair of $\Pi$.

        We combine the trees $A$ and $B$ in a similar way as in the previous step to get the tree $R$.
        This time, there does not exist a valid labeling for the original problem $\Pi$ in $R$.
        This is the tree $G_\algoa$ we aimed to construct.
        In particular, algorithm~$\algoa$ must fail to produce a valid labeling for tree $G_\algoa$.
    \end{enumerate}

    As a final remark, the number of trees constructed in step~\ref{step:rooted-tree-superlogarithmic-tree-construction} is large enough such that, no matter what labels the algorithm uses, we can always construct two trees for collection~$C_1$.
\end{proof}

\subsection{Equivalence in the sub-logarithmic region}
\label{ssec:rooted-tree-equivalence-sub}

We conclude this section by proving \cref{thm:rooted-tree-sublogarithmic}.
We prove the theorem by showing that the existence of a locality-$o(\log n)$ \onlineLOCAL algorithm for solving the LCL problem $\Pi$ implies that there exists a \emph{certificate for $O(\log* n)$ solvability} for $\Pi$.
This in turn implies that problem $\Pi$ is solvable with locality $O(\log* n)$ in the \LOCAL model~\cite{balliu21rooted-trees}.

\begin{definition}[certificate for $O(\log* n)$ solvability \cite{balliu21rooted-trees}]
    \label{def:rooted-tree-logstar-certificate}
    Let $\Pi = (\delta, \Gamma, C)$ be an LCL problem. A certificate for $O(\log* n)$ solvability of $\Pi$ consists of labels $\Gamma_{\mathcal{T}} = \{\gamma_1, \ldots, \gamma_t\} \subseteq \Gamma$, a depth pair $(d_1, d_2)$ and a pair sequences $\mathcal{T}^1$ and $\mathcal{T}^2$ of $t$~labeled trees such that
    \begin{enumerate}
        \item The depths $d_1$ and $d_2$ are coprime.
        
        \item Each tree of $\mathcal{T}^1$ (resp. $\mathcal{T}^2$) is a complete $\delta$-ary tree of depth $d_1 \ge 1$ (resp. $d_2 \ge 1$).
        
        \item Each tree is labeled by labels from $\Gamma$ correctly according to problem $\Pi$.
        
        \item Let $\bar{\mathcal{T}}^1_i$ (resp. $\bar{\mathcal{T}}^2_i$) be the tree obtained by starting from $\mathcal{T}^1_i$ (resp. $\mathcal{T}^2_i$) and removing the labels of all non-leaf nodes. It must hold that all trees $\bar{\mathcal{T}}^1_i$ (resp. $\bar{\mathcal{T}}^2_i$) are isomorphic, preserving the labeling. All the labels of the leaves of $\bar{\mathcal{T}}^1_i$ (resp. $\bar{\mathcal{T}}^2_i$) must be from set $\Gamma_{\mathcal{T}}$.

        \item The root of tree $\mathcal{T}^1_i$ (resp. $\mathcal{T}^2_i$) is labeled with label $\gamma_i$.
    \end{enumerate}
\end{definition}

In the proof, we will construct a large set of complete $\delta$-ary trees and label the nodes that are in the middle of the tree.
We call these nodes the \emph{middle nodes}.
\begin{definition}[middle nodes]
    Let $G$ be a complete $\delta$-ary tree with depth $2d$.
    The middle nodes of $G$ are the nodes of $G$ that are at distance $d$ from the root and the leaves of $G$.
\end{definition}

We are now ready construct a certificate.
\begin{proof}[Proof of \cref{thm:rooted-tree-sublogarithmic}]
    Let $\Pi = (\delta, \Gamma, C)$ be an LCL problem, and let $\algoa$ be an \onlineLOCAL algorithm solving $\Pi$ with locality $T(n) = o(\log n)$.

    We use algorithm~$\algoa$ to construct a certificate for $O(\log* n)$ solvability as follows:
    \begin{enumerate}
        \item Let $n$ be large enough to satisfy $(\delta^{T(n) + 2} + |\Gamma|)(\delta^{2T(n) + 3} - 1) \ll n$.
        Construct $\delta^{T(n) + 2} + |\Gamma|$ complete $\delta$-ary trees with depth $2T(n) + 2$.
        Each one of the tree has $\delta^{2T(n) + 3} - 1$ nodes.

        \item Use algorithm~$\algoa$ to label the middle nodes of each tree.
        Because the depth of each tree is larger than the visibility diameter of the algorithm, the algorithm does not see the roots or the leaves of the trees.
        Let $\Gamma_\mathcal{T} = \{\gamma_1, \ldots, \gamma_t\}$ the set of labels the algorithm used to label these nodes.

        \item Divide the trees into two sets $\mathcal{L}$ and $\mathcal{U}$.
        Set $\mathcal{U}$ is the minimal set such that for every label of $\Gamma_\mathcal{T}$, there exists a tree with that label in $\mathcal{U}$.
        Note that the size of set $\mathcal{U}$ is at most $|\Gamma|$, and hence the size of $\mathcal{L}$ is at least $\delta^{T(n) + 2}$.
        Order the trees in set $\mathcal{L}$ in a consistent order.

        \item We can now use the sets $\mathcal{U}$ and $\mathcal{L}$ to construct the individual trees $\mathcal{T}^1_i$ and $\mathcal{T}^2_i$.

        Consider each $\gamma_i$ in $\Gamma_\mathcal{T}$ one by one.
        Find a tree $U$ from set $\mathcal{U}$ having a node $u$ labeled with $\gamma_i$; by the construction of $\mathcal{U}$, such a tree must exist.
        Consider the subtree of $U$ rooted at $u$, and denote that tree by $U_u$.
        Tree $U_u$ has $\delta^{T(n) + 1}$ leaves.
        Identify those leaves with the roots of the first $\delta^{T(n) + 1}$ trees of $\mathcal{L}$.

        Now the root of subtree $U_u$ has label $\gamma_i$, and the nodes at depth $2T(n) + 2$ have also been labeled with labels from set $\Gamma_\mathcal{T}$.
        We can use algorithm~$\algoa$ to label all the nodes in between.
        As the result, we get that the depth-$(2T(n) + 2)$ subtree of $U_u$ is a complete labeled tree with depth $2T(n) + 2$; let it be $\mathcal{T}^1_i$.

        We can construct the trees $\mathcal{T}^2_i$ in an analogous way.
        The only difference is that instead of identifying the roots of the first $\delta^{T(n) + 1}$ trees of $\mathcal{L}$ with the leaves of the subtree $U_u$, we make the first $\delta^{T(n) + 2}$ trees of $\mathcal{L}$ the children of those leaves.
        Again, we use algorithm~$\algoa$ to label the nodes in between.
        This time we get a depth-$(2T(n) + 3)$ labeled tree as $\mathcal{T}^2_i$.

        It is easy to verify that the sequence of trees $\mathcal{T}^1$ and $\mathcal{T}^2$ actually form a valid certificate for $O(\log* n)$ solvability for problem~$\Pi$.
        In particular, the depths $2T(n) + 2$ and $2T(n) + 3$ of the trees are coprime, the leaves of each tree in both sets are labeled similarly using labels from set $\Gamma_\mathcal{T}$, and for every label of set $\Gamma_\mathcal{T}$, there exists a tree in both sets having root labeled with that label.
    \end{enumerate}

    The fact that the existence of a locality-$o(\log n)$ \onlineLOCAL algorithm implies the existence of a certificate for $O(\log* n)$ solvability is sufficient to prove the theorem.
    In particular, if no certificate existed, we know that the locality of the problem in the \onlineLOCAL model must be $\Omega(\log n)$.
\end{proof}

\end{document}